\documentclass[reqno, 12pt, a4paper]{amsart}

\usepackage{amsmath}
\usepackage{amsthm}
\usepackage{amssymb}
\usepackage{amsfonts}
\usepackage{mathtools}

\usepackage{fullpage}

\usepackage{dsfont}

\usepackage{mathrsfs}

\usepackage[dvipsnames,svgnames]{xcolor}
\colorlet{MyBlue}{DodgerBlue!60!Black}
\colorlet{MyGreen}{DarkGreen!85!Black}

\definecolor{ngreen}{RGB}{56, 188, 83}
\definecolor{nred}{RGB}{196, 39, 39}

\usepackage{subfigure}
\usepackage{tikz}
\usetikzlibrary{calc,patterns}
\usetikzlibrary{arrows}
\usetikzlibrary{arrows.meta}
\usetikzlibrary{decorations.markings}
\usepackage{diagbox}
\usepackage{longtable}

\usepackage{acronym}
\usepackage{booktabs}       
\usepackage{latexsym}
\usepackage{nicefrac}       
\usepackage{wasysym}
\usepackage{xspace}
\usepackage[shortlabels]{enumitem}

\usepackage[authoryear,comma]{natbib}

\usepackage{hyperref}
\hypersetup{
colorlinks=true,
linktocpage=true,
pdfstartview=FitH,
breaklinks=true,
pdfpagemode=UseNone,
pageanchor=true,
pdfpagemode=UseOutlines,
plainpages=false,
bookmarksnumbered,
bookmarksopen=false,
bookmarksopenlevel=1,
hypertexnames=true,
pdfhighlight=/O,
urlcolor=MyBlue!60!black,linkcolor=MyBlue!70!black,citecolor=DarkGreen!70!black, 
pdftitle={},
pdfauthor={},
pdfsubject={},
pdfkeywords={},
pdfcreator={pdfLaTeX},
pdfproducer={LaTeX with hyperref}
}

\numberwithin{equation}{section}  
\usepackage[sort&compress,capitalize,nameinlink]{cleveref}
\crefname{app}{Appendix}{Appendices}

\crefrangeformat{equation}{\upshape(#3#1#4)\textendash(#5#2#6)}

\usepackage[textwidth=30mm]{todonotes}

\newcommand{\debug}[1]{#1}

\usepackage[xcolor]{changes}
\definechangesauthor[name=Ben, color=blue]{Ben}
\definechangesauthor[name=Andrea, color=red]{Andrea}
\definechangesauthor[name=Marco, color=magenta]{Marco}
\definechangesauthor[name=Ziwen, color=orange]{Ziwen}

\theoremstyle{plain}
\newtheorem{theorem}{Theorem}

\newtheorem*{corollary*}{Corollary}
\newtheorem{lemma}[theorem]{Lemma}
\newtheorem{proposition}[theorem]{Proposition}

\theoremstyle{definition}
\newtheorem{definition}[theorem]{Definition}
\newtheorem*{definition*}{Definition}

\newtheorem*{hypothesis*}{Hypothesis}

\theoremstyle{remark}
\newtheorem{remark}{Remark}
\newtheorem*{remark*}{Remark}
\newtheorem*{notation*}{Notational remark}

\numberwithin{theorem}{section}
\numberwithin{remark}{section}
\numberwithin{example}{section}

\newcommand{\R}{\mathbb{R}}

\newcommand{\N}{\mathbb{N}}

\newcommand{\game}{\debug \Gamma}

\newcommand{\play}{\debug i}
\newcommand{\playalt}{\debug j}
\newcommand{\nplayers}{\debug N}
\newcommand{\players}{\debug {[\nplayers]}}

\newcommand{\act}{\debug s}
\newcommand{\actalt}{\debug t}
\newcommand{\actprof}{\boldsymbol{\act}}
\newcommand{\actprofalt}{\boldsymbol{\actalt}}
\newcommand{\actprofr}{\boldsymbol{\debug r}}
\newcommand{\actprofu}{\boldsymbol{\debug u}}
\newcommand{\actprofw}{\boldsymbol{\debug w}}

\newcommand{\actprofzero}{\boldsymbol{\debug 0}}

\newcommand{\actions}{\debug S}

\newcommand{\equi}[1]{\debug \Theta^{#1}}
\newcommand{\equii}[1]{\debug \Theta_{#1}}

\newcommand{\rgame}{\debug \Xi}

\newcommand{\payoff}{\debug g}

\newcommand{\eq}[1]{#1^{\ast}}

\newcommand{\brn}{\br^{(\nplayers)}}

\DeclareMathOperator{\equilibria}{\mathsf{\debug{PNE}}}
\DeclareMathOperator{\sequilibria}{\mathsf{\debug{SPNE}}}

\DeclareMathOperator{\br}{\mathsf{\debug{BR}}}
\DeclareMathOperator{\BRD}{\mathsf{\debug{Brd}}}

\newcommand{\hyperc}{\mathscr{\debug H}}
\newcommand{\vertices}{\mathscr{\debug V}}
\newcommand{\verticesalt}{\mathscr{\debug U}}
\newcommand{\edges}{\mathscr{\debug E}}
\newcommand{\boundary}{\debug \partial}

\newcommand{\edge}{\debug e}

\newcommand{\orient}[1]{\overrightarrow{#1}}
\newcommand{\rorient}[1]{\overleftarrow{#1}}

\newcommand{\orientedges}[2]{\debug{\braces{#1\to#2}}}

\newcommand{\profitable}{\mathscr{\debug D}}
\newcommand{\expproc}{\mathscr{\debug Q}}
\newcommand{\percol}{\mathscr{\debug B}}

\newcommand{\comp}{\debug C}
\newcommand{\ecomp}{\debug \mu}
\newcommand{\ncomp}{\debug j}

\newcommand{\card}[1]{\mathsf{\debug{card}}\parens*{#1}}

\newcommand{\event}{\debug A}
\newcommand{\Expect}{\mathsf{\debug E}}
\newcommand{\probeq}{\debug \alpha}
\newcommand{\probless}{\debug \beta}
\newcommand{\Prob}{\mathsf{\debug P}}
\newcommand{\Probalt}{\mathsf{\debug Q}}

\DeclareMathOperator{\Bernoulli}{\mathsf{\debug{Bernoulli}}}
\DeclareMathOperator{\Binomial}{\mathsf{\debug{Binomial}}}
\DeclareMathOperator{\cardd}{\mathsf{\debug{card}}}

\DeclareMathOperator{\Norm}{\mathsf{\debug{Norm}}}
\DeclareMathOperator{\Poisson}{\mathsf{\debug{Poisson}}}
\DeclareMathOperator{\TV}{\mathsf{\debug{TV}}}

\newcommand{\diff}{\, \textup{d}}
\newcommand{\dist}{\debug \rho}
\newcommand{\bigoh}{\debug O}

\DeclareMathOperator*{\disjcup}{\dot\cup}
\DeclareMathOperator{\expo}{e}
\DeclareMathOperator*{\even}{\mathsf{\debug{even}}}
\DeclareMathOperator*{\odd}{\mathsf{\debug{odd}}}

\newcommand{\ball}{\debug B}

\newcommand{\bernoullipar}{\debug \theta}
\newcommand{\Bessel}{\debug I}
\newcommand{\consta}{\debug K}

\newcommand{\fairperc}[1]{\widetilde{#1}}
\newcommand{\floorm}{\debug m}

\newcommand{\hamm}{\debug h}
\newcommand{\incid}{\debug Y}

\newcommand{\isoln}{\debug \Upsilon}
\newcommand{\kolmo}{\debug \kappa}
\newcommand{\larg}{\mathscr{\debug L}}

\newcommand{\mixturepar}{\debug \gamma}
\newcommand{\nolarg}{\mathscr{\debug M}}

\newcommand{\normCDF}{\debug \Phi}

\newcommand{\poissonpar}{\debug \lambda}
\newcommand{\radius}{\debug r}
\newcommand{\rand}{\debug Z}

\newcommand{\run}{\debug k}

\newcommand{\subs}{\debug I}

\newcommand{\threeincid}{\debug G}
\newcommand{\threeinter}{\debug H}
\newcommand{\thresh}{\debug \ell}
\newcommand{\trap}{\mathscr{\debug T}}
\newcommand{\upbound}{\debug \zeta}

\DeclarePairedDelimiter{\braces}{\{}{\}}
\DeclarePairedDelimiter{\bracks}{[}{]}
\DeclarePairedDelimiter{\parens}{(}{)}

\DeclarePairedDelimiter{\abs}{\lvert}{\rvert}

\DeclarePairedDelimiter{\ceil}{\lceil}{\rceil}
\DeclarePairedDelimiter{\floor}{\lfloor}{\rfloor}

\DeclarePairedDelimiterX{\braket}[2]{\langle}{\rangle}{#1,#2}

\DeclarePairedDelimiterX{\inner}[2]{\langle}{\rangle}{#1,#2}
\DeclarePairedDelimiterX{\setdef}[2]{\{}{\}}{#1:#2}

\DeclarePairedDelimiterXPP{\probof}[1]{\prob}{(}{)}{}{%

#1}

\DeclarePairedDelimiterXPP{\exof}[1]{\ex}{[}{]}{}{%

#1}

\newacro{BRD}{best-response dynamics}
\newacro{CLT}{central limit theorem}
\newacro{PNE}{pure Nash equilibrium}
\newacro{NE}{Nash equilibrium}
\newacroplural{NE}[PE]{Nash equilibria}
\newacro{PNE}{pure Nash equilibrium}
\newacroplural{PNE}[PNE]{pure Nash equilibria}
\newacro{SPNE}{strict pure Nash equilibrium}
\newacroplural{SPNE}[SPNE]{strict pure Nash equilibria}
\newacro{WHP}{with high probability}
\newacro{WVHP}{with very high probability}

\title{Pure Nash Equilibria and Best-Response Dynamics in Random Games}

\author{Ben Amiet}
\address{School of Mathematical Sciences, Monash  University, Melbourne, Australia}
\email{ben.amiet@monash.edu}
\author{Andrea Collevecchio}
\address{School of Mathematical Sciences, Monash  University, Melbourne, Australia}
\email{Andrea.Collevecchio@monash.edu}
\author{Marco Scarsini}
\address{Department of Economics and Finance, Luiss, Viale Romania 32, 00196 Rome, Italy}
\email{marco.scarsini@luiss.it}	
\author{Ziwen Zhong}
\address{School of Mathematical Sciences, Monash  University, Melbourne, Australia}
\email{ziwen.zhong@monash.edu}



\begin{document}

\begin{abstract}
In finite games, mixed Nash equilibria always exist, but pure equilibria may fail to exist.
To assess the relevance of this nonexistence, we consider  games where the payoffs are drawn at random.
In particular, we focus on games where a large number of players can each choose one of two possible strategies and the payoffs are i.i.d.\ with the possibility of ties. 	
We provide asymptotic results about the random number of \aclp{PNE}, such as fast growth and a central limit theorem, with bounds for the approximation error. 
Moreover, by using a new link between percolation models and game theory, we describe in detail the geometry of \aclp{PNE} and show that, when the probability of ties is small, a best-response dynamics reaches a \acl{PNE} with a probability that quickly approaches one as the number of players grows.  
We show that various phase transitions depend only on a single parameter of the model, that is, the probability of having ties.
\end{abstract}

\maketitle


\section{Introduction}


\subsection{Background and motivation}

A  \acfi{PNE}\acused{PNE} in a normal form game is a profile of strategies (one for each player) such that, given the choice of the other players, no player has an incentive to make a different choice. In other words, deviations from an equilibrium are not profitable for any player.
This concept of equilibrium, although quite simple and powerful, has the drawback that not every game admits \acp{PNE}.
The idea of \emph{mixed strategy}---i.e., a  probability distribution over a player's strategy set---goes back to \citet{Bor:CRASP1921} and \citet{vNe:MA1928}.
John Nash showed that, if mixed strategies are allowed and the choice criterion is the expected payoff with respect to their product, then any game with a finite number of players and strategies admits an equilibrium \citep{Nash:PNAS1950,Nash:AM1951}.
As before, a mixed equilibrium is a profile of mixed strategies that does not allow profitable deviations.

Although the definition and properties of mixed strategies and mixed equilibria are clear, their interpretation is far from unanimous. 
Section~3.2 of \citet{OsbRub:MITPress1994}, dedicated to  the interpretation of mixed equilibria, has paragraphs individually signed by each of the two authors, since they could not reach an agreement.
In general, \acp{PNE} have a stronger epistemic foundation than mixed equilibria. 
As mentioned before, the main problem of \acp{PNE} is existence. 

Some authors have tried to frame this problem in a stochastic way: given a set of players and a set of strategies for each player, if payoffs are drawn at random, what is the probability that the game admits \acp{PNE}?
More precisely, what is the distribution of the number of \acp{PNE} in a game with random payoffs?
The answer to this question clearly depends on the way random payoffs are drawn.
In any case, for a fixed number of players and strategies, the answer is computationally daunting.
It is for this reason that some papers have chosen to investigate the problem from an asymptotic viewpoint; that is, they have looked at the limit distribution of the number of \acp{PNE} as the number of either strategies or players increases to infinity.
	
The basic common assumption of much of this literature is that the distribution of the random payoffs is nonatomic and  payoff profiles are independent.
Under these hypotheses, the probability that two payoffs coincide is zero and, as a consequence, calculations are significantly simplified.

It is well-known that \ac{PNE} are hard to compute
\citep{DasGolPap:SIAMJC2009}.
One way to address the issue is to devise iterative procedures that converge to a \ac{PNE}. 
For instance, some adaptive procedures start from a strategy profile and allow a single player (picked \emph{at random}) to choose a different strategy.
In better-response dynamics the chosen player picks uniformly among the  strategies that guarantee a higher payoff than the one the player presently  has, given the strategies of the other players. 
In \acfi{BRD}\acused{BRD} the player chooses  one of the strategies  that  give  the highest possible payoff  (there could be ties),  given the strategies of the other players. When each player has only two available strategies, the two procedures coincide.
In either scheme, if no such strategy exists, the player will not move and a different player is chosen at random.
When a new strategy profile is reached, the process is repeated. 
If we reach a profile for which no player has a profitable deviation, then the process has reached a \ac{PNE}.
The question is whether, starting from any strategy profile, a \ac{PNE} is reached. 
In general the answer is negative: first, because a game may fail to have \acp{PNE}; second, due to the fact that even when \acp{PNE} exist, players may be trapped in a subset of vertices that does not contain a \ac{PNE}.
One way to determine the severity of this failure to reach a \ac{PNE} via \acl{BRD} is to examine games with random payoffs.


\subsection{Our contribution} 

In the present paper we consider games in which the number $\nplayers$ of players is large, each player has two strategies, and payoffs are random. 
The main novelty of our approach is to show the strict relation between  games with random payoffs and percolation, and to provide analytic results rather than simulations.
In games with random payoffs, a significant part of the existing literature has focused on the number of \acp{PNE} when the payoffs are i.i.d.\ with absolutely continuous distribution  with respect to the Lebesgue measure.
Here we extend this analysis  to the case where ties are allowed. We also discuss the behavior of \ac{BRD}. 
In particular, we provide results that concern not only the number of \acp{PNE}, but also how easily they can be reached via \ac{BRD}. 
The main tool for this analysis is a correspondence between a random oriented graph that represents our random game and a suitable percolation graph. 

A main feature of the present work is that we dispense with the assumption of nonatomic distribution of the payoffs and therefore allow ties to exist. 
We  show that the probability of ties plays a crucial role in many ways. 
For example, it  determines the asymptotic distribution of the number of \acp{PNE}.
Moreover, we  use tools from percolation theory to describe the geometry of the set of \acp{PNE}, which also depends on the probability of ties.
This description permits us to analyze the performance of \ac{BRD}.
This has been extensively done in the literature for the class of potential games. 
In our paper we can show that, asymptotically in the number of players, with high probability \ac{BRD} converges to a \ac{PNE}, if the probability of ties is positive, but small (less than $0.68$).  

As mentioned before, the probability of ties in the payoffs, which we call $\probeq$, is the fundamental parameter in this model.  
Different values of $\probeq$ produce different possible behaviors in the number of \acp{PNE}. 
We will show that, for as long as $\probeq$ is positive, the game has typically around $(1+\probeq)^{\nplayers}$ \acp{PNE}.
For all $\nplayers$ large enough and $\probeq$  small enough we have that \ac{BRD} converges to a \ac{PNE} (\cref{th:accessibility,th:BRD} below). 
Moreover, for all large $\nplayers$, if $\probeq$ is strictly less than $1/2$, then all \ac{PNE} are reachable via \ac{BRD} from any deterministic starting point. 
Conversely, some of them are unreachable when $\probeq$ is at least $1/2$.
Furthermore, when $\probeq$ is positive, \cref{th:case-atoms} shows a concentration of the number of \ac{PNE} around $(1+\probeq)^{\nplayers}$  and establishes a \ac{CLT}. 
After a draft of this paper was ready, we learned that a similar \ac{CLT} was obtained in \cite{Rai:P7YSM2003}. Our approach is different and relies on a Poisson approximation, which we believe is of independent interest.

To illustrate this phenomenon, we plot in \cref{fi:arbre1} the case where the payoffs take only the values $\{-1, 1\}$ with equal probability (notice that $\probeq = 0.5$ in this context).  
The average number of \ac{PNE} exactly fits the curve $(1.5)^\nplayers$, confirming our prediction.
\begin{figure}[ht]
	\centering
	\includegraphics[width=8cm]{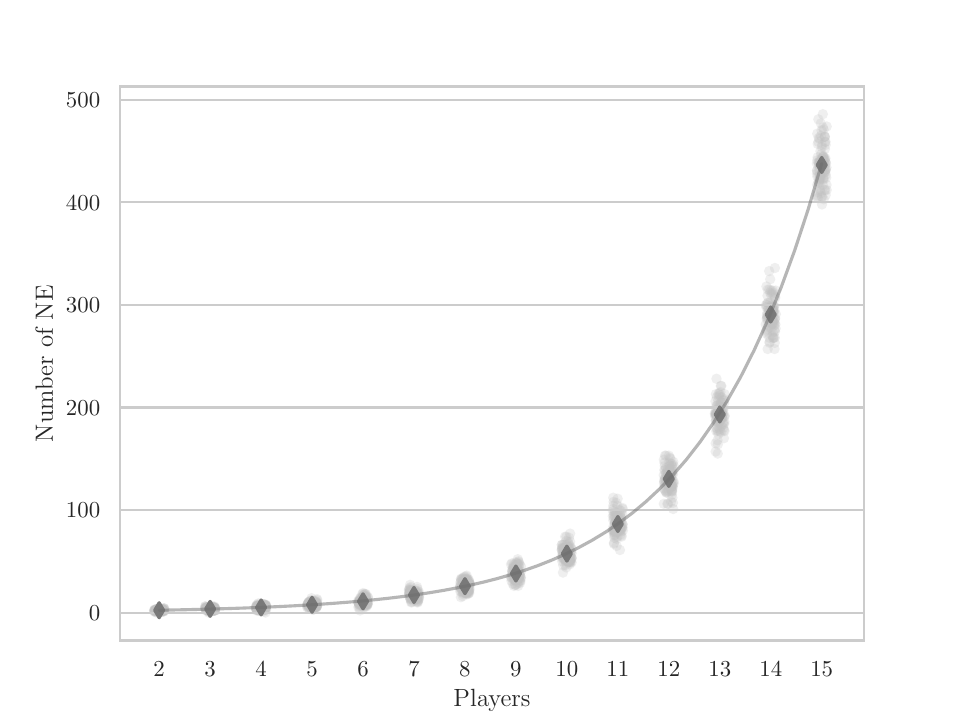}
	\caption{\label{fi:arbre1} Number of \acp{PNE} for $2\leq\nplayers\leq 15$, $\probeq = 0.5$, with 100 trials per $\nplayers$. Diamond markers represent average number per value of $\nplayers$, and the curve $(1.5)^{\nplayers}$ is included for comparison.}
\end{figure}
\begin{figure}[ht]
	\centering
	\includegraphics[width=8cm]{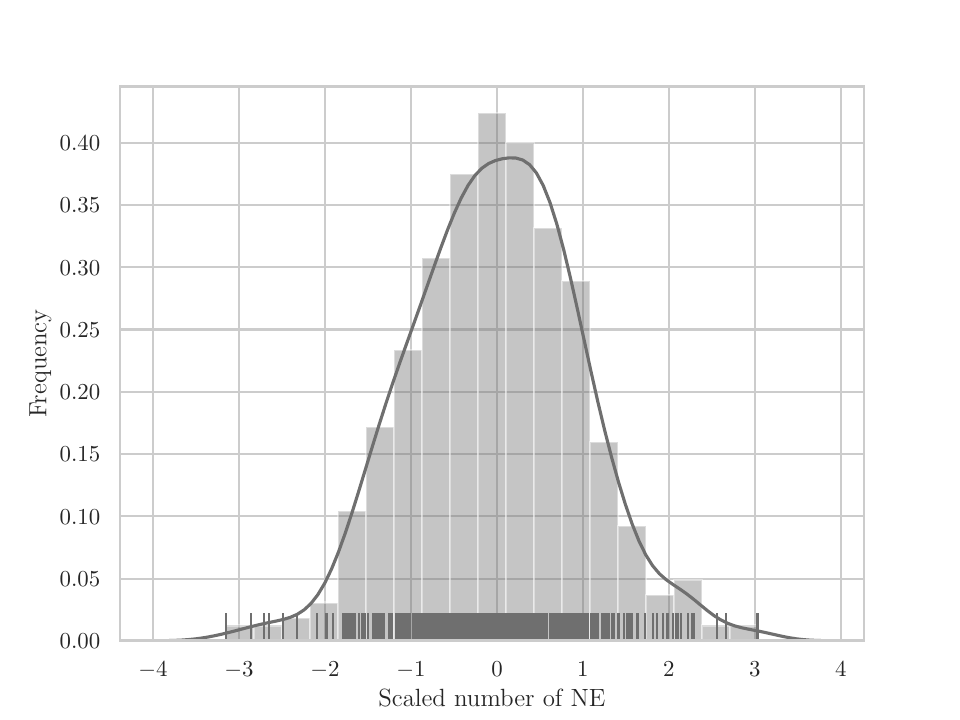}
	\caption{\label{fi:arbre2} \ac{CLT} result for $\nplayers=15$, $\probeq=0.9$, with 500 trials.}
\end{figure}
Moreover, we are able to quantify the fluctuations (see \cref{th:case-atoms}) which are of the order $(1+\probeq)^{\nplayers/2}$. 
Finally, the number of \acp{PNE}, once properly rescaled, rapidly converges to a standard normal (see \cref{fi:arbre2}).
We emphasize that our results depend on the payoff distributions only through the parameter $\probeq$, and remain applicable even when the payoff distributions vary across players.


\subsection{Related work}

As mentioned before, several papers have considered aspects related to the number of \acp{PNE} in games with random payoffs. 
In many of the papers that we consider below, and unless otherwise stated, the random payoffs are i.i.d.\ from a continuous distribution.

\citet{Gol:AMM1957} considered zero-sum two-person games  and showed that the probability of having a \ac{PNE} goes to zero as the number of strategies grows.
He also briefly mentioned the case of payoffs with a Bernoulli distribution.
\citet{GolGolNew:JRNBSB1968} considered general two-person games and showed that the probability of having at least one \ac{PNE} converges to $1-\expo^{-1}$ as the number of strategies diverges.
\citet{Dre:JCT1970} generalized this result to the case of an arbitrary finite number of players.

More recent papers have looked at the asymptotic distribution of the number of \acp{PNE}.
\citet{Pow:IJGT1990} showed that, when the number of strategies of at least two players goes to infinity, the distribution of the number of \acp{PNE} converges to a $\Poisson(1)$.

\citet{Sta:GEB1995} derived an exact formula for the distribution of the number of \acp{PNE} in random games and obtained the result in \citet{Pow:IJGT1990} as a corollary.
\citet{Sta:MOR1996} dealt with the case of two-person symmetric games and obtained Poisson convergence for the number of both symmetric and asymmetric \acp{PNE}.
In all the above models, the expected number of \acp{PNE} is in fact $1$.
Under different hypotheses, this expected number diverges. 
For instance, \citet{Sta:MSS1997,Sta:EL1999} showed that this is the case for games with vector payoffs and for games of common interest, respectively.
In \citet{Sta:EL1999} both strictly and weakly ordinal preferences were studied. 
\citet{Rai:P7YSM2003} used the Chen-Stein method to bound the distance between the distribution of the normalized number of \acp{PNE} and a normal distribution. His result  does not assume continuity of the payoff distributions.

\citet{RinSca:GEB2000} weakened the hypothesis of i.i.d.\ payoffs; that is, they assumed that  payoff vectors corresponding to different strategy profiles are i.i.d., but they allowed some dependence within the same payoff vector. 
In this setting, they proved asymptotic results when either the number of players or the number of strategies diverges. 
More precisely, if each payoff vector has a multinormal exchangeable distribution with correlation coefficient $\rho$, then the following hold: for $\rho$ negative the number of \acp{PNE} goes to zero in probability, for $\rho=0$ it converges to a $\Poisson(1)$, and for $\rho$ positive it diverges and a \ac{CLT} holds.

\citet{Tak:GEB2008} considered the distribution of the number of \acp{PNE} in a random game with two players, conditionally  on the game having nondecreasing best-response functions.
This assumption greatly increases the expected number of \acp{PNE}.
\citet{DasDimMos:AAP2011} extended the framework of games with random payoffs to graphical games. Strategy profiles are vertices of a graph and players' strategies are binary, like in our model. 
Moreover, their payoff depends only on their strategy and the strategies of their neighbors.
The authors studied how the structure of the graph affects existence of \acp{PNE} and they examined both deterministic and random graphs.

The issue of solution concepts in games with random payoffs has been explored by various authors in different directions.
For instance, \citet{Coh:PNAS1998} studied the probability that Nash equilibria (both pure and mixed) in a finite random game maximize the sum of the players' payoffs.
\citet{PeiTak:GEB2019} devoted their attention to rationalizable strategies in two-person games with random payoffs and performed an asymptotic analysis in the number of strategies.

Finding a \ac{PNE} in a game is PPAD-complete (\citet{DasGolPap:SIAMJC2009}). 
Therefore, given this computational difficulty, several learning procedures have been proposed to reach an equilibrium by playing the game  several times  \citep[see, e.g.,][]{TarVaz:AGT20017,BluMan:AGT2007}.
Probably the simplest such procedure is \ac{BRD}. 
This approach has been taken, among others, by \citet{Blu:GEB1993,You:E1993,FriMez:JET2001,TakYam:EB2002} and \citet{FabJagSha:TCS2013}.
The main problem that arises is that \ac{BRD} is guaranteed to converge to a \ac{PNE} only when the game is of some specific type, for instance, a potential game (\citet{MonSha:GEB1996}).
\citet{Ber:E1984} introduced the concept of point rationalizable strategies, i.e., strategies that survive iterated elimination of never best responses against pure strategies.

The performance of \ac{BRD} in randomly drawn potential games has been studied in \citet{CouDurGauTou:NetGCoop2014,DurGau:AGT2016} and \citet{DurGarGau:PE2019}.
To be able to deal also with games for which \ac{BRD} does not converge to a \ac{PNE}, \citet{GoeMirVet:FOCS2005} defined the concept of \emph{sink equilibria}  (which in this paper are called traps).
A trap  is a strongly connected set of two or more vertices with no edges leading out of the set.
If players are selected at random and asked to choose a best response, the process may eventually reach a sink equilibrium and wander on the game's components forever.
\citet{ChrMirSid:TCS2012} considered a similar model,  focusing on the rate of convergence to approximate solutions of the game. 
\citet{DutKes:SODA2017} considered \ac{BRD} in the context of combinatorial auctions.

The idea of generating games at random to check properties of learning procedures was used by \citet{GalFar:PNAS2013} and, more recently, by \citet{PanHeiFar:SA2019}, who studied---mainly through simulations---the behavior of various  learning procedures in games whose payoffs are drawn at random from a multinormal distribution.

Some of our results are proved by using a connection with percolation theory, a field introduced by
\citet{BroHam:PCPS1957}. 
Since then, the theory has developed very quickly and has become very important in both the mathematics and physics communities. 
For a general account on percolation, see \cite{Gri:Springer1999}. 
We will focus on  percolation on the hypercube, and will use classical results by \cite{ErdSpe:CMA1979} and \cite{Bol:CUP2001}, as well as more recent results by \citet{McDScoWhi:arXiv2020}.
The connection between percolation and random oriented graphs has been studied, for instance, in an unpublished manuscript by \citet{Lin:arXiv2009}, in the case where the graph is fully oriented, i.e., there are no ties in the payoffs.


\subsection{Organization of the paper}

\cref{se:number} deals with the number of \acp{PNE} in a random game.
\cref{se:BRD} studies the behavior of \ac{BRD} in these games.
The interaction between games with random payoffs and percolation is expounded in \cref{se:percolation}.
\cref{se:proofs} contains the proofs.


\section{Number of \aclp{NE} in random games}
\label{se:number}

We first introduce some notation that will be adopted throughout the paper.	 
We consider a game 
\begin{equation}\label{eq:game-def}
\game=\parens*{\players,(\actions_{\play})_{\play\in\players},(\payoff_{\play})_{\play\in\players}},
\end{equation} 
where $\players\coloneqq\braces{1,\dots,\nplayers}$ is the set of players and $\actions_{\play}$ is the set of strategies of each player $\play\in\players$. 
We set $\actions=\times_{\play\in\players}\actions_{\play}$, and we let $\payoff_{\play}\colon\actions\to\R$ be the payoff function of player $\play$. For each $\actprof=\parens{\act_{1},\dots,\act_{\nplayers}}\in\actions$, we call $\actprof_{-\play}$ the strategy profile of all players except $\play$.

\begin{definition}\label{de:Nash-equilibrium}
	A strategy profile $\eq{\actprof}$ is a \ac{PNE} of the game $\game$ if for all $\play\in\players$ and for all $\act_{\play}\in\actions_{\play}$ we have
	\begin{equation}\label{eq:Nash}
	\payoff_{\play}(\eq{\actprof})\ge\payoff_{\play}(\act_{\play},\eq{\actprof}_{-\play}).
	\end{equation}
	The set of \acp{PNE} of $\game$ is denoted by $\equilibria(\game)$.

	A strategy profile $\eq{\actprof}$ is a \acfi{SPNE}\acused{SPNE} if  for all $\play\in\players$ and for all $\act_{\play}\in\actions_{\play}$ we have
	\begin{equation}\label{eq:strict-Nash}
	\payoff_{\play}(\eq{\actprof}) >\payoff_{\play}(\act_{\play},\eq{\actprof}_{-\play}).
	\end{equation}
	The set of \acp{SPNE} of  $\game$  is denoted by $\sequilibria(\game)$.
\end{definition}

Our goal is to examine some generic properties of games with binary strategies and a large number of players. 
To achieve this, we assume that the payoffs of our game are drawn at random. 
To this end, consider a probability space $(\Omega,\mathscr{F},\Prob)$, on which the following sequence of random games is defined. Let $\rgame_{\nplayers}$ be a game with $\nplayers$ players, $\actions_{\play}=\braces{0,1}$ for each $\play\in\players$, and random i.i.d.\ payoffs.
In particular, for each $\actprof\in\actions$, the payoff $\payoff_{\play}(\actprof)$ is the realization of a random variable $\rand_{\play}^{\actprof}$, and the random variables $\parens{\rand_{\play}^{\actprof}}_{\play\in\players,\actprof\in\actions}$ are i.i.d.. 
The symbol $\rand$ denotes a generic independent copy of $\rand_{\play}^{\actprof}$.

We also define 
\begin{equation}\label{eq:probeq} 
\probeq\coloneqq\Prob\parens*{\rand_{1}=\rand_{2}}, \quad 
\probless\coloneqq\Prob\parens*{\rand_{1}<\rand_{2}}=\frac{1-\probeq}{2},
\end{equation}
where $\rand_{1}$ and $\rand_{2}$ are i.i.d.\ copies of $\rand$.
As we will see, all of the results in the paper will depend on $\probeq$. 
Most of the existing literature deals with the case $\probeq=0$.
In what follows we will show the role of $\probeq$ in various phase transitions.
In particular, we are interested in the asymptotic behavior of $\equilibria(\rgame_{\nplayers})$ and $\sequilibria(\rgame_{\nplayers})$.

For any set $\event$, the symbol $\card{\event}$ denotes its cardinality.
Observe that the probability with which a strategy profile $\actprof$ is a \ac{PNE} is  $\parens{1-\probless}^{\nplayers}$, and the probability that it is an \ac{SPNE} is  $\probless^{\nplayers}$.
As a consequence, for any $\nplayers \ge 1$ we have
\begin{align}
\label{eq:exact-number-PNE}
\Expect\bracks*{\card{\equilibria(\rgame_{\nplayers})}}
&=2^{\nplayers}\parens{1-\probless}^{\nplayers},\\
\label{eq:exact-number-SNE}
\Expect\bracks*{\card{\sequilibria(\rgame_{\nplayers})}}
&=2^{\nplayers}\probless^{\nplayers}.
\end{align}
This implies that the expected number of \acp{PNE} is always $1$ when $\probeq=0$ and diverges when $\probeq>0$. 
The following theorems provide sharp description of the asymptotic behavior of the number of \acp{PNE} and \acp{SPNE}.

\begin{theorem}
[Behavior of \aclp{SPNE}]
\label{th:SNE} 
Consider a sequence of random games $\rgame_{\nplayers}$.
\begin{enumerate}[\upshape(a)]
    \item
\label{it:th:SNE-1}     
    If $\probeq = 0$, then, for all $\run \in \N\cup\{0\}$, we have
\begin{equation}
\label{eq:SNE-alpha=0}
\lim_{\nplayers\to\infty}\Prob\parens*{\card{\sequilibria(\rgame_{\nplayers})} = \run} = \frac{\expo^{-1}}{\run!}.
\end{equation}
	\item 
\label{it:th:SNE-2}
If $\probeq>0$, then 
\begin{equation}
\label{eq:SNE-alpha>0}
\Prob\parens*{\lim_{\nplayers\to\infty}\card{\sequilibria(\rgame_{\nplayers})} = 0}=1.
\end{equation}
	\end{enumerate}
\end{theorem}

\begin{remark} \label{re:NE}
	 When $\probeq=0$,  the numbers of \acp{PNE} and of \acp{SPNE} are almost surely equal, as any two payoffs are almost surely different.
	In this case, convergence of the number of \acp{PNE} to a Poisson distribution as the number of players increases was proved by  \citet{ArrGolGor:AP1989,RinSca:GEB2000} for any fixed number of strategies.
	
	Conversely, when $\probeq>0$,  the number of both \acp{PNE} and \acp{SPNE} have radically different behavior. This fact will be better described in \cref{th:case-atoms,th:accessibility} below.
\end{remark}

Call $\normCDF$ the cumulative distribution function of a standard normal random variable.
\begin{theorem}
[\ac{CLT} for \aclp{PNE}]
\label{th:case-atoms}
	Assume that the law of $\rand$ has atoms, i.e., $\probeq>0$. 
	Then there exists a constant $\consta_\probeq>0$ which depends only on $\probeq$, such that 
\begin{equation}\label{eq:atoms-wvhp}
	\sup_x \abs*{\Prob\parens*{\frac{\card{\equilibria(\rgame_{\nplayers})} -  (1+\probeq)^{\nplayers}}{(1+\probeq)^{\nplayers/2}} \le x} - \normCDF(x)} \le \consta_{\probeq} \nplayers \max\parens*{\frac{\parens{1+\probeq}}{2},\frac{1}{\parens{1+\probeq}^{1/2}}}^{\nplayers}.
\end{equation}
\end{theorem}

\begin{remark} 
Call $\Norm(\mu,\sigma^{2})$ the normal distribution with mean $\mu$ and variance $\sigma^{2}$.
Define the Kolmogorov distance $\kolmo$ between two probability measures $\Prob,\Probalt$ on $\R$ as 
\begin{equation}
\label{eq:Kolmogorov}
\kolmo(\Prob,\Probalt) \coloneqq \sup_{x\in\R}\abs*{\Prob\parens*{(-\infty,x]}-\Probalt\parens*{(-\infty,x]}}.
\end{equation}
With an abuse of language, we will speak of Kolmogorov distance between two random variables to indicate the distance between their corresponding laws.
Hence, \cref{eq:atoms-wvhp} is equivalent to 
	\begin{equation}\label{eq:atoms-wvhp-2}
	\kolmo\parens*{\card{\equilibria(\rgame_{\nplayers})}, \Norm((1+\probeq)^{\nplayers},(1+\probeq)^{\nplayers})}
	\le \consta_{\probeq}\nplayers \max\parens*{\frac{\parens{1+\probeq}}{2},\frac{1}{\parens{1+\probeq}^{1/2}}}^{\nplayers}.
	\end{equation}
\end{remark}

A direct consequence of \cref{th:case-atoms} is that the number of \acp{PNE} grows exponentially in $\nplayers$, whenever $\probeq>0$. 
The following is  a more precise statement.

\begin{theorem}
\label{th:geo}
If $\probeq>0$, then
\begin{equation}
\Prob\parens*{\lim_{\nplayers \to \infty} \frac{\card{\equilibria(\rgame_{\nplayers})}}{(1+\probeq)^{\nplayers}}=1}=1.
\end{equation}
\end{theorem}


\section{\acl{BRD}}
\label{se:BRD}
We focus on games with $\nplayers$ players and $2$ possible strategies per player. 
We associate to any such game $\game$ the graph $\hyperc_{\nplayers}=\parens{\vertices_{\nplayers},\edges_{\nplayers}}$, where the set of vertices is the set of strategy profiles, i.e., $\vertices_{\nplayers}=\actions$,  and two vertices $\actprof,\actprofalt$ are connected by an edge in $\edges_{\nplayers}$ if and only if 
\begin{equation}\label{eq:neighbors}
\act_{\play}\neq\actalt_{\play}\text{ for exactly one }\play\in\players\text{ and }\act_{\playalt}=\actalt_{\playalt}\text{ for all }\playalt\neq\play.
\end{equation}
In this case we write $\actprof\sim_{\play}\actprofalt$. Moreover, we write  $\actprof\sim\actprofalt$ if  $\actprof\sim_{\play}\actprofalt$ for some $\play\in\players$ and say that $\actprof,\actprofalt$ are \emph{neighbors}. For each pair $\actprof,\actprofalt$  of neighbors, call $\bracks{\actprof,\actprofalt}$ the edge connecting them.  
The vertex $(0,0, \ldots, 0) \in \vertices_{\nplayers}$ is denoted by $\actprofzero$.

\begin{definition} \label{de:nash-enash} 
Given a strategy profile $\actprof\in\vertices_{\nplayers}$, define
\begin{equation}
\label{eq:profitable}
\profitable(\actprof) \coloneqq \{\actprofalt \in \vertices_{\nplayers}  \colon \exists\; \play\in\players \text{ such that } \actprof\sim_{\play}\actprofalt  \text{ and } \payoff_{\play}(\actprofalt)>\payoff_{\play}(\actprof)\}
\end{equation}
to be the set of strategy profiles that are \emph{strictly profitable deviations} from  $\actprof$.

We define recursively  a discrete-time Markov chain $\BRD(\game, \actprof)= (\BRD_{\run}(\game, \actprof))_{\run\in \N \cup \{0\}}$ on $\vertices_{\nplayers}$ as follows. 
The process starts at $\actprof \in \vertices_{\nplayers}$,  i.e., $\BRD_{0}(\game, \actprof) = \actprof$.
At time $\run+1$, independently of the history of the process,  pick an element uniformly at random from the random set $\profitable\left(\BRD_{\run}(\game, \actprof)\right)$, say $\actprofw$, and set $\BRD_{\run+1}(\game, \actprof)= \actprofw$.
Define 
\begin{equation}\label{eq:BRN}
\brn \coloneqq \BRD(\rgame_{\nplayers}, \actprofzero).
\end{equation} 
Finally, we say that $\BRD(\game, \actprof)$ converges to a \ac{PNE} if there exists some $\run_0\in\N$ such that, for all $\run\geq\run_0$, $\profitable\left(\BRD_{\run}(\game, \actprof)\right) = \varnothing$.
\end{definition}

We associate to each random game $\rgame_{\nplayers}$ a partially oriented random graph $\orient{\hyperc}_{\nplayers}^{\probless}=\parens{\vertices_{\nplayers},\orient{\edges}_{\nplayers}}$ as follows. 
Let $\actprof\sim_{\play}\actprofalt$;
then the oriented edge $\overrightarrow{\bracks{\actprof,\actprofalt}}$ from $\actprof$ towards $\actprofalt$ is in $\orient{\edges}_{\nplayers}$ if and only if   
$\rand^{\actprof}_{\play}<\rand^{\actprofalt}_{\play}$.  
An unoriented edge connects $\actprof$   and $\actprofalt$ in $\orient{\hyperc}_{\nplayers}^{\probless}$ 
if an only if $\rand^{\actprof}_{\play}=\rand^{\actprofalt}_{\play}$.
A similar representation has been used by \citet{You:E1993} and \citet{CanMenOzdPar:MOR2011} in the setting of deterministic games.
Obviously, in our context,  where the game $\rgame_{\nplayers}$ has random payoffs, the partial orientation of $\orient{\hyperc}_{\nplayers}^{\probless}=\parens{\vertices_{\nplayers},\orient{\edges}_{\nplayers}}$ is itself random.
If the law of $\rand$ is nonatomic, then the probability that two payoffs coincide is zero. 
Therefore, $\orient{\hyperc}_{\nplayers}^{1/2}$ is a `proper'  random orientation of $\hyperc_{\nplayers}$, where each edge is independently oriented in one direction or the other with  probability $1/2$.
If, on the other hand, the law of $\rand$ has atoms, then $\Prob\parens{\rand^{\actprof}_{\play}=\rand^{\actprofalt}_{\play}}>0$, so with positive probability some edges have no orientation. 

\begin{definition}\label{de:accessible}
	We say that $\actprofalt$ is \emph{directly accessible} from $\actprof$ if the oriented edge  $\overrightarrow{\bracks{\actprof,\actprofalt}}\in\orient{\edges}_{\nplayers}$. 
	We say that  $\actprofalt$ is \emph{accessible} from $\actprof$ if there exists a finite sequence $\actprof_{0},\actprof_{1},\dots,\actprof_{\run}$ such that $\actprof=\actprof_{0}$, $\actprofalt=\actprof_{\run}$ and, for all $\play\in\braces{0,\dots,\run-1}$, we have $\overrightarrow{\bracks{\actprof_{\play},\actprof_{\play+1}}}\in\orient{\edges}_{\nplayers}$. 
\end{definition}

Notice that $\actprofalt$ is directly accessible from $\actprof$ if and only if $\actprofalt$ is a profitable deviation from $\actprof$ for some player $\play$. 
\cref{de:accessible} has a natural interpretation in terms of $\brn$. 
Recall that this process starts at $\actprofzero$. 
Then $\actprofalt$ is accessible from $\actprofzero$ if and only if there is a positive probability that
$\brn$ reaches $\actprofalt$.

An example with three players is given in \cref{fi:game3}, where the orientation of the edges is induced by the payoffs in the accompanying table. 
The absence of ties produces a complete orientation of the hypercube.
The black edges are the possible paths of a \ac{BRD} starting at the vertex $(0,0,0)$.
The two vertices represented as filled circles, i.e., $(1,1,1)$ and $(0,0,1)$, are the two \acp{PNE} of the game.
Each of them can be  reached, with positive probability, by $\brn$.
\cref{th:BRD} below proves that for all large enough $\nplayers$ this is always the case when  $\probeq$ is small enough.

\begin{figure}[ht]
	\centering
	\begin{tabular}{|c|c|c|}
		\hline
		\multicolumn{3}{|c|}{Player 3 - Strategy  $0$}\\ 
		\hline
		\diagbox{Player 1}{Player 2} & $0$ & $1$\\ 
		\hline
		$0$ & $(0.949,0.530,0.055)$ & $(0.110,0.567,0.794)$ \\
		\hline
		$1$ & $(0.199,0.097,0.319)$ & $(0.202,0.549,0.174)$ \\
		\hline
		\hline
		\multicolumn{3}{|c|}{Player 3 - Strategy $1$}\\ 
		\hline
		\diagbox{Player 1}{Player 2} & $0$ & $1$\\  
		\hline
		$0$ & $(0.815,0.774,0.508)$ & $(0.292,0.684,0.126)$ \\ 
		\hline
		$1$ & $(0.209,0.659,0.569)$ & $(0.542,0.709,0.426)$ \\
		\hline
	\end{tabular}
    
    \vspace{1em}
    
	\begin{tikzpicture}[scale=.5]
    	\scriptsize
    	
    	\node [label={below left:$(0,0,0)$}] (0) at (0,0) {};
    	\node [label={below right:$(1,0,0)$}] (1) at (4,0) {};
    	\node [label={above left:$(0,1,0)$}] (2) at (0,4) {};
    	\node [label={right:$(1,1,0)$}] (3) at (4,4) {};
    	\node [label={left:$(0,0,1)$}] (4) at (2.55,2.55) {};
    	\node [label={below right:$(1,0,1)$}] (5) at (6.55,2.55) {};
    	\node [label={above left:$(0,1,1)$}] (6) at (2.55,6.55) {};
    	\node [label={above right:$(1,1,1)$}] (7) at (6.55,6.55) {};
    	
    	\draw (0,0) circle (3pt);
    	\draw (4,0) circle (3pt);
    	\draw (0,4) circle (3pt);
    	\draw (4,4) circle (3pt);
    	\filldraw (2.55,2.55) circle (3pt);
    	\draw (6.55,2.55) circle (3pt);
    	\draw (2.55,6.55) circle (3pt);
    	\filldraw (6.55,6.55) circle (3pt);
    	
    	\begin{scope}[thick,decoration={markings, mark=at position 0.33 with {\arrow{latex}}, mark=at position 0.66 with {\arrow{latex}}}] 
        	\draw[postaction={decorate}] (1) -- (0)[color=Silver];
        	\draw[postaction={decorate}] (0) -- (2);
        	\draw[postaction={decorate}] (0) -- (4);
        	\draw[postaction={decorate}] (1) -- (3)[color=Silver];
        	\draw[postaction={decorate}] (1) -- (5)[color=Silver];
        	\draw[postaction={decorate}] (2) -- (3);
        	\draw[postaction={decorate}] (6) -- (2)[color=Silver];
        	\draw[postaction={decorate}] (3) -- (7);
        	\draw[postaction={decorate}] (5) -- (4)[color=Silver];
        	\draw[postaction={decorate}] (6) -- (4)[color=Silver];
        	\draw[postaction={decorate}] (5) -- (7)[color=Silver];
        	\draw[postaction={decorate}] (6) -- (7)[color=Silver];
    	\end{scope}
	\end{tikzpicture}
	\caption{Representation of $\game_{3}$ on $\braces{0,1}^{3}$.}
	\label{fi:game3}		
\end{figure}
Our next result shows the existence of a sharp phase transition in the accessibility of \acp{PNE}. 
Roughly speaking, as the mass of the atoms in the distribution of $\rand$ grows, so too does the number of \acp{PNE}, though some may not be accessible from $\actprofzero$.
Hence, in this case, some \acp{PNE} may not be reachable by $\brn$.
Consider the partition 
\begin{equation}
\label{eq:partition}
\vertices_{\nplayers}=\orient{\larg}_{\nplayers}^{\probless,\actprofzero}\disjcup\orient{\nolarg}_{\nplayers}^{\probless,\actprofzero}
\end{equation}
in such a way that $\orient{\larg}_{\nplayers}^{\probless,\actprofzero}$ is the set that contains $\actprofzero$ as well as all vertices $\actprofalt$ that are accessible from $\actprofzero$ in the oriented graph $\orient{\hyperc}_{\nplayers}^{\probless}$. 
Conversely, all vertices that are \emph{not} accessible from $\actprofzero$ are contained in $\orient{\nolarg}_{\nplayers}^{\probless,\actprofzero}$.

\begin{theorem}
[Accessibility of \aclp{PNE}]
\label{th:accessibility}
Let $\probeq = 1 - 2 \probless$ as defined  in \cref{eq:probeq}.
	\begin{enumerate}[label={\upshape{(\alph*)}}]
		\item\label{it:th:accessibility-1}
		If $0\le\probeq<1/2$, then 
		\begin{equation}
		\label{eq:alpha<1/2}
		\sum_{\nplayers=1}^{\infty}\Prob\parens*{\equilibria(\rgame_{\nplayers})
		\not\subseteq\orient{\larg}_{\nplayers}^{\probless,\actprofzero}}<\infty.
		\end{equation}
	
		\item\label{it:th:accessibility-2}
		If $\probeq=1/2$, then 
\begin{equation}
\label{eq:alpha=1/2}
\lim_{\nplayers\to\infty}
\Prob\parens*{\card{\equilibria(\rgame_{\nplayers}) \cap \orient{\nolarg}_{\nplayers}^{\probless,\actprofzero}}>0}\ge 1-\expo^{-1}.
\end{equation}		
	
\item\label{it:th:accessibility-3}
If $\probeq>1/2$, then, for any $\consta>0$,
\begin{equation}
\label{eq:alpha>1/2}
\lim_{\nplayers\to\infty}
\Prob\parens*{\card{\equilibria(\rgame_{\nplayers}) \cap
\orient{\nolarg}_{\nplayers}^{\probless,\actprofzero}} > \consta} = 1.
\end{equation}		
	\end{enumerate}
\end{theorem}

This theorem can be interpreted as follows:

\noindent
\ref{it:th:accessibility-1} for the case $0\le\probeq<1/2$, we can use the first Borel-Cantelli lemma and conclude that there exists a finite random $\nplayers^*$ such that for all $\nplayers \ge \nplayers^*$ each of the \ac{PNE} in $\rgame_{\nplayers}$ is potentially reachable by $\brn$;

\noindent
\ref{it:th:accessibility-2} for $\probeq=1/2$, with positive probability there exist \acp{PNE} that are not reachable by $\brn$; 

\noindent
\ref{it:th:accessibility-3} for $\probeq>1/2$, the number of \acp{PNE} that are not reachable by $\brn$ grows to infinity, with probability approaching $1$.

Therefore, we have an interesting phase transition at $\probeq=1/2$.

To complete the picture, we give a result about the convergence of $\brn$. 
We use the notation $\floor{x}$ for the greatest integer less than or equal to $x$ and $\ceil{x}$ for the least integer greater than or equal to $x$.
The following theorem shows that $\brn$ converges to a \ac{PNE} if $\probeq$ is positive but smaller than a threshold.

\begin{theorem}
[Convergence of $\brn$]
\label{th:BRD} 
	If $\probeq>0$ satisfies
	\begin{equation}\label{eq:convergence}
	\floor*{- \frac 1{\log_2{(1/2 + \probeq/2})}} \le 3,
	\end{equation}
i.e., if $0<\probeq<2^{3/4}-1\approx 0.68$, 
then 
\begin{equation}
\label{eq:P-not-converge-NE}
\sum_{\nplayers=1}^{\infty} \Prob\left(\brn \text{ does not converge to a \ac{PNE}}\right) < \infty. 
\end{equation}
\end{theorem}

Using the first Borel-Cantelli lemma we  conclude that there exists a finite random $\nplayers^*$ such that for all $\nplayers \ge \nplayers^*$ the process $\brn$ converges to a \ac{PNE}. 
In the proof, we  use the fact  that $\brn$ fails to converge to a \ac{PNE} if and only if it enters a trap, which is defined as follows.

\begin{definition}
\label{de:trap}
An oriented graph is \emph{strongly connected} if every vertex is accessible from every other vertex.
A \emph{trap} is a strongly connected subgraph $\orient{\trap}$ of  $\orient{\hyperc}_{\nplayers}^{\probless}$ with two or more vertices,  such that, for all $\actprof \in \vertices(\orient{\trap})$ and all $\actprofalt \notin \vertices(\orient{\trap})$, we have that $\actprofalt$ is not accessible from $\actprof$, where $\vertices(\orient{\trap})$ is  the vertex set of $\orient{\trap}$. 
\end{definition}

For any trap $\orient{\trap}$, the following holds:
\begin{equation}\label{eq:trap}
\Prob\parens*{\exists j \ge \run \text{ such that } \brn_{j}  \notin \vertices(\orient{\trap}) \mid  \brn_{\run} \in  \vertices(\orient{\trap})} = 0.
\end{equation}

Notice that a trap contains at least four vertices, by a simple parity argument, which is described as follows.
Define the Hamming distance $\hamm \colon \actions\times \actions$ as
\begin{equation}
\label{eq:Hamming}
\hamm(\actprof,\actprofalt)\coloneqq \cardd\braces*{\play \colon \act_{\play}\neq\actalt_{\play}},
\end{equation}
which, for binary vectors is equal to $\sum_{\play\in\players}\abs{\act_{\play}-\actalt_{\play}}$. 
Each vertex $\actprof$ of  $\hyperc_{\nplayers}$ is called even if $\hamm(\actprof, \actprofzero)$ is even, and is called odd otherwise.  
If $\actprof$ is even, all its neighbors are odd, and vice versa.  
Hence, the smallest cycles in this graph are of length $4$. 
As a trap must include   a cycle because it is strongly connected in $\orient{\hyperc}_{\nplayers}^{\probless}$, a trap contains at least four vertices.

The simulation result in \cref{fi:traps}  illustrates the phenomenon described by \cref{th:BRD} for $\probeq<0.68$.

\begin{figure}[ht]
	\centering
	\includegraphics[width=10cm]{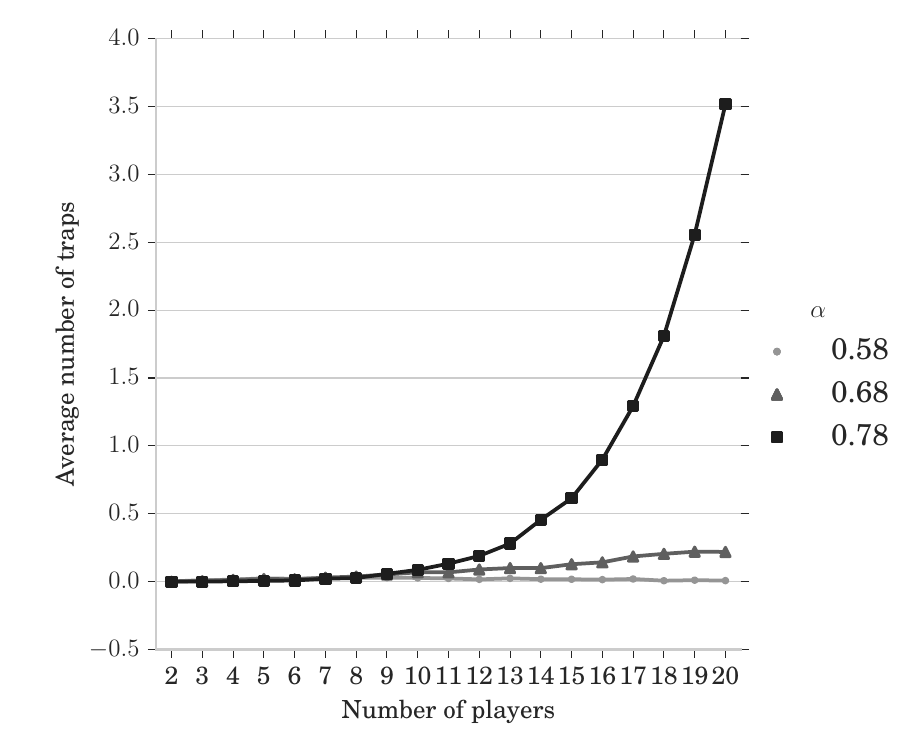}
	\caption{\label{fi:traps} Average number of traps found in games ranging from $\nplayers=2$ to 20 and 
$\probeq = 2^{3/4}-1.1,2^{3/4}-1,2^{3/4}-0.9$, with 1500 trials per combination of $\probeq$ and $\nplayers$.}
\end{figure}

\begin{remark}
\label{re:Poisson}
This paper deals with \acp{PNE} of games in normal form. 
The definition of this concept requires the knowledge of players' preferences over strategy profiles. 
For any two profiles $\actprof,\actprofalt$, with $\actprof\sim_\play\actprofalt$, player $\play$ either prefers $\actprof$ to $\actprofalt$ or prefers $\actprofalt$ to $\actprof$ or is indifferent between the two. 
Any numerical representation of these preferences is invariant with respect to strictly increasing transformations. 
When dealing with random games, given the symmetry induced by the hypothesis of i.i.d.\ payoffs, what is relevant to determine the behavior of \acp{PNE} is the probability that a certain preference exists between two strategy profiles that the player can choose from.
This implies  that payoff distributions that are quite different lead to the same behavior of \acp{PNE} if they have the same parameter $\probeq$. Consider the following examples.

\begin{itemize}
\item 
If the distribution of $\rand$ is a mixture of an absolutely continuous distribution and a degenerate distribution with weights $(1-\mixturepar)$ and $\mixturepar$, respectively, then $\probeq=\mixturepar^{2}$.
\item If $\rand\sim\Bernoulli(\bernoullipar)$, then $\probeq=\bernoullipar^{2}+(1-\bernoullipar)^{2}$. 
\item If $\rand\sim\Poisson(\poissonpar)$, then 
\begin{equation}\label{eq:alpha-Poisson}
\probeq
=\sum_{\run=0}^{\infty}\frac{\expo^{-2\poissonpar}\poissonpar^{2\run}}{(\run!)^{2}}
=\expo^{-2\poissonpar}\Bessel_{0}(2\poissonpar),
\end{equation}
where $\Bessel_{0}$ is the modified Bessel function of the first kind
\citep[Eq.~2, p.~77]{Wat:CUP1995}.
\end{itemize}

For each value of $\bernoullipar\in[0,1]$, there exist $\mixturepar$ and $\poissonpar$ that produce the same value of $\probeq$.
Notice that  if $\rand\sim\Bernoulli(\bernoullipar)$, then the possible values of $\probeq$ are the interval  $[1/2, 1]$, whereas both the mixture and the Poisson distribution can cover the whole interval $(0,1]$ of values of $\probeq$. 
In the case of mixture, also the case $\probeq =0$ can be covered, and the proof follows directly from the definition of the model.\\
For the Poisson case, we reason as follows.  Using the integral representation of the modified Bessel function \citep[Eq.~4, p.~181]{Wat:CUP1995}, we obtain
\begin{equation}
    \frac{\textup{d}}{\textup{d}\poissonpar}\expo^{-2\poissonpar}\Bessel_0(2\poissonpar) = \frac{1}{\pi}\frac{\textup{d}}{\textup{d}\poissonpar}\int_0^\pi \expo^{2\poissonpar(\cos(t)-1)}\diff t = \frac{1}{\pi}\int_0^\pi 2(\cos(t)-1)\expo^{2\poissonpar(\cos(t)-1)}\diff t.
\end{equation}
Clearly the contents of this integral are negative for $0<\poissonpar<\pi$, so $\expo^{-2\poissonpar}\Bessel_{0}(2\poissonpar)$ is decreasing in $\poissonpar$. If $\poissonpar=0 $, then $\probeq = 1$, while for $\poissonpar$ diverging to infinity we have that $\probeq$ converges to $0$. By a simple continuity argument we see  that  all values of $\probeq$ can be covered, and to each $\probeq \in (0,1]$ there exists exactly one value of $\poissonpar$ satisfying \cref{eq:alpha-Poisson}.

\begin{figure}[ht]
	\centering 
	\includegraphics[width=8cm]{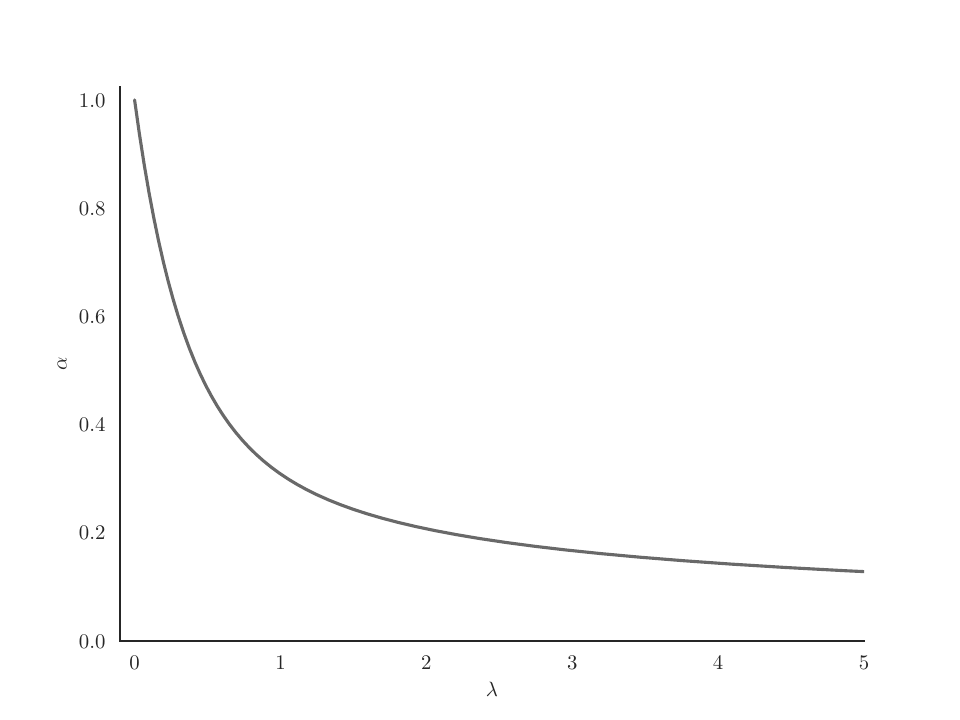}
	\caption{\label{fi:alpha-theta} $\probeq$ as a function of $\poissonpar$.}
\end{figure}
\end{remark}


\section{Bond percolation}
\label{se:percolation}

Most of  the  proofs of the results in \cref{se:number,se:BRD} require some tools from percolation theory. The main step to establish the connection between random games and percolation is to construct a coupling  with the following property. The set of strategies that are  accessible by  $\brn$ coincides with the  connected component containing $\actprofzero$ of a percolation on the hypercube. 
We then use results about the geometry of this cluster to infer limit theorems for the number of Nash equilibria and the \ac{BRD}.
Our results rely on   \citet{McDScoWhi:arXiv2020}. 
Whenever we mention a connected component of the percolation that has a certain property, we refer to the \emph{largest} {connected} component  that satisfies that property.
Moreover, we use the term \emph{giant} component to refer to the connected component with the  largest  number of vertices. 
Of course, this definition makes sense only when such a component is unique, which is the case in the sequel, for all $\nplayers$  large enough. 

Independent  bond percolation on $\hyperc_{\nplayers}$ is defined as follows.  
For each edge in $\hyperc_{\nplayers}$, flip an independent coin having probability $\probless$ of showing heads. 
If the toss shows heads, then declare the edge to be \emph{open}; otherwise the edge is \emph{closed}. 
The random subgraph of $\hyperc_{\nplayers}$ obtained by deleting the closed edges is called percolation; it includes all vertices in $\vertices_{\nplayers}$, but could be disconnected.  
This model allows us to give a detailed description of the geometry of \acp{PNE}.

The next result relates $\orient{\larg}_{\nplayers}^{\probless,\actprofzero}$ to a  percolation on the hypercube.

\begin{proposition}
\label{pr:percolation}
For any $\probless \in [0, 1/2]$ there exists a  percolation  $\percol^{\probless}_{\nplayers}$ that satisfies the following property: the vertex set of its connected component that contains $\actprofzero$, denoted by $\larg^{\probless,\actprofzero}_{\nplayers}$,  coincides with $\orient{\larg}_{\nplayers}^{\probless,\actprofzero}$.  
\end{proposition}

\begin{proof}
First, we define the event
	\begin{equation}\label{eq:ind-s-to-t}
	\orientedges{\actprofr}{\actprofalt}\coloneqq\braces{\overrightarrow{\bracks{\actprofr,\actprofalt}}\in\orient{\edges}_{\nplayers}}.
	\end{equation}
	Since each player has only two strategies, we have that $\orientedges{\actprofr}{\actprofalt}$ is independent of $\orientedges{\actprofu}{\actprofw}$ for every $\orientedges{\actprofu}{\actprofw}\notin \big\{\orientedges{\actprofr}{\actprofalt}, \orientedges{\actprofalt}{\actprofr}\big\}$.
	
	For any subset $\verticesalt \subseteq \vertices_{\nplayers}$, we call $\orient{\boundary \verticesalt}$  the set of vertices  in  $\verticesalt^{c} = \vertices_{\nplayers}\setminus \verticesalt $ 
	that are out-neighbors of some elements in $\verticesalt$, that is,
\begin{equation}\label{eq:out-neigh}
\orient{\boundary\verticesalt}:= \braces*{\actprofw \in  \verticesalt^{c} \colon \exists\,\actprofu\in\verticesalt\text{ such that }\orientedges{\actprofu}{\actprofw}\text{ is true}}
\end{equation}
	and  we call $\boundary \verticesalt$  the set of vertices in $\verticesalt^{c}$ that are neighbors of some elements in $\verticesalt$, that is,  
\begin{equation}\label{eq:neigh}
\boundary\verticesalt :=\braces*{\actprofw \in  \verticesalt^{c} \colon \exists\,\actprofu\in\verticesalt\text{ such that }  \actprofu \sim \actprofw}.
\end{equation}

 Define  the process 
	\begin{equation}\label{eq:Pk+1}
	\expproc_{1}=\braces{\actprofzero}\text{ and, for each }\run\in\N,\
	\expproc_{\run+1}=\expproc_{\run}\cup\orient{\boundary\expproc}_{\run} .
	\end{equation}
We also need to construct a finite sequence of unoriented random graphs such that each graph of the sequence is a bond percolation with parameter $\probless$. 
The last percolation  in this finite sequence has the property that we want, that is, the vertex set  of its  connected component   containing $\actprofzero$  is equal to $\orient{\larg}_{\nplayers}^{\probless,\actprofzero}$.
Start with a bond percolation on $\hyperc_{\nplayers}$ with parameter $\probless$ that is independent of the random variables $\braces*{\rand^{\actprofalt}_{\play}\colon \play\in\players,\actprofalt\in\vertices_{\nplayers}}$. 
Call the resulting graph $\percol_{1}$.	For every $\run\ge 1$ we will update $\percol_{\run}$ by changing the status of some edges at each stage,  in such a way that  $\percol_{\run+1}$ is still a bond percolation with parameter $\probless$. For each edge $\edge \in \edges_{\nplayers}$, we define $\percol_{\run}\braces{\edge}$ the status (open or closed) of edge $\edge$ in $\percol_{\run}$. 
	We obtain  $\percol_{\run+1}$ from  $\percol_{\run}$,  by updating \emph{all and only} the edges  in  $\edges_{\nplayers}$ that connect an element of $\expproc_{\run}$ to an element of $\boundary \expproc_{\run}$. 
	More precisely, for any $\actprofu \in \expproc_{\run}$ and any  $\actprofw \in \boundary  \expproc_{\run}$, with $\actprofu \sim \actprofw$, set
	\begin{equation}\label{eq:Bk+1}
	\percol_{\run+1}\braces{[\actprofu, \actprofw]} = 
	\begin{cases}
	\text{open,} & \text{if }\orientedges{\actprofu}{\actprofw},\\
	\text{closed,} & \text{otherwise};
	\end{cases}
	\end{equation}
	for all other edges $\edge \in \edges_{\nplayers}$, we have $\percol_{\run+1}\braces{\edge} = \percol_{\run}\braces{\edge}$.
	
	Since the status of edges is updated independently of the original configuration and with i.i.d.\ $\Bernoulli(\probless)$ random variables, we have that $\percol_{\run + 1}$ is still a bond percolation with parameter $\probless$.
	
	Notice that, in the worst-case scenario, each of these processes explores the entirety of  $\hyperc_{\nplayers}$ in $2^{\nplayers}$ iterations. 
	That is $\percol_{\run+1} = \percol_{\run}$ and $\expproc_{\run+1}=\expproc_{\run}$ for all $\run \ge 2^{\nplayers}$.
	By construction, $\expproc_{2^{\nplayers}}$  is exactly the set of vertices in the  connected component  of  the percolation graph $\percol_{2^{\nplayers}}$ that contains $\actprofzero$.
	In this context, we have $\orient{\larg}_{\nplayers}^{\probless,\actprofzero} =\expproc_{2^{\nplayers}}$.  Set $\percol^{\probless}_{\nplayers}= \percol_{2^{\nplayers}}$.
	Recall that $\larg^{\probless,\actprofzero}_{\nplayers}$ is the set of vertices in the  connected component  that contains $\actprofzero$ in the percolation graph $\percol^{\probless}_{\nplayers}$. 
By construction, we have $\larg^{\probless,\actprofzero}_{\nplayers} =  \orient{\larg}_{\nplayers}^{\probless,\actprofzero}$. 
\end{proof}

The case $\probless=1/2$ of \cref{pr:percolation} was studied by \citet[Lemma~2.1]{Lin:arXiv2009} in an unpublished manuscript, using different techniques.

\begin{remark}
\label{re:coupling}

We constructed a percolation $\percol^{\probless}_{\nplayers}$ that is coupled with the random partial orientation of the cube, which in turn is induced by a random game. 
This will allow us to identify vertices in the percolation  $\percol^{\probless}_{\nplayers}$ with reference to properties of both the random partially oriented hypercube and the random game.
For instance, we will talk about \acp{PNE} in the percolation. 
\end{remark}

Define  $\larg^{\probless}_{\nplayers}$ to be the set of vertices in the giant component  of the percolation $\percol^{\probless}_{\nplayers}$ introduced in \cref{pr:percolation}.
We call the complement of the giant component, denoted by $\nolarg^{\probless}_{\nplayers}$, the \emph{fragment of the percolation}.
With an abuse of notation, we use $\nolarg^{\probless}_{\nplayers}$ also to denote the set of vertices  of the fragment, i.e., $\vertices_{\nplayers}\setminus\larg^{\probless}_{\nplayers}$.

\begin{proposition}
\label{pr:finally}  
For any  $\probless \in (0, 1/2]$,   we have 
\begin{equation}\label{eq:L=L}
\sum_{\nplayers=1}^\infty 
\Prob\parens*{\larg^{\probless}_{\nplayers} \neq \orient{\larg}^{\probless,\actprofzero}_{\nplayers}}<\infty.
\end{equation}
\end{proposition}

\begin{proof}
We  first prove the proposition for  $\probless \in (0, 1/2)$. 
In virtue of \cref{pr:percolation}, we have  $\orient{\larg}_{\nplayers}^{\probless,\actprofzero} =\larg^{\probless,\actprofzero}_{\nplayers}$, and in order to prove \cref{eq:L=L} it is enough to prove 
\begin{equation}
\label{eq:L=L1}
\sum_{\nplayers=1}^\infty 
\Prob\parens*{\larg^{\probless}_{\nplayers} \neq \larg^{\probless,\actprofzero}_{\nplayers}}<\infty.
\end{equation}
By symmetry, we have
\begin{equation}
\Prob\parens*{\actprofzero \in \nolarg_{\nplayers}^{\probless} \mid \card{\nolarg_{\nplayers}^{\probless}}} = \frac{\card{\nolarg_{\nplayers}^{\probless}}}{2^{\nplayers}}.
\end{equation}
Hence,
\begin{equation}\label{eq:clar1}
\Prob\parens*{\larg^{\probless}_{\nplayers} \neq \larg^{\probless,\actprofzero}_{\nplayers}\mid \card{\nolarg_{\nplayers}^{\probless}}} 
=\Prob\parens*{\actprofzero \in \nolarg_{\nplayers}^{\probless} \mid \card{\nolarg_{\nplayers}^{\probless}}}  = \frac{\card{\nolarg_{\nplayers}^{\probless}}}{2^{\nplayers}}.
\end{equation}

Given two sequences $(a_n)_n$ and $(b_n)_n$, we  say that $a_n = \bigoh(b_n)$ if $ \sup_{n \in \N} a_n/b_n <\infty$. Moreover, we say $a_n = o(b_n)$  if  $\lim_{n \to \infty} a_n/b_n  =0$.  
We make use of the following lemma about upper and lower bounds for the cardinality of the fragment. 

\begin{lemma}[\protect{\citet[Theorem~1(a)]{McDScoWhi:arXiv2020}}]
\label{le:McD1}
For each $\varepsilon > 0$ and for a fixed $\probless\in(0, 1/2)$,  the percolation $\percol_{\nplayers}^{\probless}$ satisfies
\begin{equation}\label{eq:borr1}
    \sum_{\nplayers=1}^\infty \Prob\parens*{\abs*{\card{\nolarg_{\nplayers}^{\probless}} - \Expect\bracks*{\card{\nolarg_{\nplayers}^{\probless}}}} \geq \varepsilon\sqrt{\nplayers (2(1-\probless))^{\nplayers}}} <\infty,
\end{equation}
where $\Expect\bracks*{\card{\nolarg_{\nplayers}^{\probless}}} =(2(1-\probless))^{\nplayers}(1 + \bigoh(\nplayers (1-\probless)^{\nplayers}))$.
\end{lemma}
Set
\begin{equation}
\label{eq:hndef}
\upbound_{\nplayers} \coloneqq \Expect\bracks*{\card{\nolarg_{\nplayers}^{\probless}}} + \varepsilon\sqrt{\nplayers (2(1-\probless))^{\nplayers}}.
\end{equation}
Using \cref{eq:borr1,eq:clar1}, we obtain
\begin{equation}
\label{eq:L=L2}
\begin{split}
\sum_{\nplayers=1}^\infty 
\Prob\parens*{\larg^{\probless}_{\nplayers} \neq \larg^{\probless, \actprofzero}_{\nplayers}} 
&\le\sum_{\nplayers=1}^\infty \Prob\parens*{\braces*{\larg^{\probless}_{\nplayers} \neq \larg^{\probless, \actprofzero}_{\nplayers}} \cap \braces*{\card{\nolarg_{\nplayers}^{\probless}} < \upbound_{\nplayers}}} \\
&\qquad + \sum_{\nplayers=1}^\infty \Prob\parens*{\card{\nolarg_{\nplayers}^{\probless}} \geq \upbound_{\nplayers}}\\
&\le \sum_{\nplayers=1}^\infty \frac{\upbound_{\nplayers}}{2^{\nplayers}} + \sum_{\nplayers=1}^\infty 
\Prob\parens*{\card{\nolarg_{\nplayers}^{\probless}} \geq  \upbound_{\nplayers}}<\infty. 
\end{split}
\end{equation}

Next, we prove the case $\probless=1/2$ of \cref{pr:finally}.
It is well known that there exists a coupling such that if $\probless < \probless'$, then $\larg^{\probless}_{\nplayers} \subseteq \larg^{\probless'}_{\nplayers}$  \citep[see, e.g.,][Theorem~2.1, page~36]{Bol:CUP2001}. 
In other words, $\larg^{\probless}_{\nplayers}$ is monotone in $\probless$. As 
\begin{equation}
\label{eq:equal-random-sets}
\larg^{\probless}_{\nplayers} \neq \larg^{\probless,\actprofzero}_{\nplayers} \Longleftrightarrow \actprofzero \notin \larg^{\probless}_{\nplayers},
\end{equation}
it follows that
\begin{equation}\label{eq:prob-monotonic}
\Prob\parens*{\larg^{\probless}_{\nplayers} \neq {\larg}^{\probless,\actprofzero}_{\nplayers}}
\end{equation}
is nonincreasing in $\probless$. 
For any $\probless \in (0,1/2)$, we have that 
\begin{equation*}
\sum_{\nplayers=1}^\infty 
\Prob\parens*{\larg^{1/2}_{\nplayers} \neq \larg^{1/2,\actprofzero}_{\nplayers}} \le 
\sum_{\nplayers=1}^\infty 
\Prob\parens*{\larg^{\probless}_{\nplayers} \neq \larg^{\probless,\actprofzero}_{\nplayers}}<\infty. \qedhere
\end{equation*}
\end{proof}

The following lemma will be used several times in our proofs.
\begin{lemma}[\protect{\citet[Theorem~2(a)]{McDScoWhi:arXiv2020}}]
\label{le:McD2}
Fix $\probless \in (0, 1/2]$. For any $\radius\in\N$ and any $\actprof\in\vertices_{\nplayers}$, call
\begin{equation}\label{eq:ball}
\ball_{\radius}(\actprof)\coloneqq\braces{\actprofalt\colon\hamm(\actprof,\actprofalt)\le\radius},
\end{equation}
where $\hamm$ is the Hamming distance. Set 
\begin{equation}\label{eq:mbeta}
\floorm_{\probless}\coloneqq\floor*{\frac{1}{-\log_2(1-\probless)}}.
\end{equation}
Then there exists $\overline{\delta}>0$ such that, for any $\delta < \overline{\delta}$, we have 
\begin{equation}
\label{eq:sum-McD}
\sum_{\nplayers=1}^{\infty}
\Prob\parens*{\exists\; \actprofalt \colon \card{{\ball_{\ceil{\delta \nplayers}}}(\actprofalt)\setminus\orient{\larg}_{\nplayers}^{\probless,\actprofzero}} > \floorm_{\probless}} < \infty.
\end{equation}
\end{lemma}

\begin{remark} Notice that \citet[Theorem~2(a)]{McDScoWhi:arXiv2020} is actually formulated in terms of giant component of  $\percol^{\probless}_{\nplayers}$, i.e., $\larg^{\probless}_{\nplayers}$ instead of $\orient{\larg}_{\nplayers}^{\probless,\actprofzero}$. 
This substitution holds in virtue  of \cref{pr:percolation} and  \cref{pr:finally}. 
Moreover, \citet[Theorem~2(a)]{McDScoWhi:arXiv2020} covers only the case $\probless \in (0, 1/2)$. 
On the other hand, the case $\probless =1/2$ follows immediately using monotonicity of percolation with respect to $\probless$.  
In fact,  choose $\probless<1/2$ such that $\floorm_{\probless} = \floorm_{1/2}=1$. 
By monotonicity, we have that $\larg^{\probless, \actprofzero}_{\nplayers}$ is stochastically smaller than $\larg^{1/2,  \actprofzero}_{\nplayers}$. 
Hence
\begin{equation}
\label{eq:prob-L-1/2-beta}
\Prob\parens*{\exists\; \actprofalt \colon \card{\ball_{\ceil{\delta \nplayers}}(\actprofalt)\setminus\orient{\larg}_{\nplayers}^{1/2,\actprofzero}} > 1}  
\le \Prob\parens*{\exists\; \actprofalt \colon \card{\ball_{\ceil{\delta \nplayers}}(\actprofalt)\setminus\orient{\larg}_{\nplayers}^{\probless,\actprofzero}} > 1}.
\end{equation}
\end{remark}

The next result  focuses on the nonatomic case, that is, $\probless = 1/2$ (or equivalently, $\probeq=0$). This corresponds to the classical bond percolation with parameter $1/2$. 
A vertex is   \emph{isolated } if it has degree zero in the graph induced by the percolation.

\citet{ErdSpe:CMA1979} analyzed the asymptotic behavior of $\percol^{1/2}_{\nplayers}$, and showed that  the random graph is connected with probability tending to one. 
Upon further inspection of their proof, it is evident that, with probability tending to one, the giant component of this percolation contains all the vertices in $\vertices_{\nplayers}$ with the exception of some isolated vertices  in the random graph $\percol^{1/2}_{\nplayers}$.
As the following result was not explicitly stated in \citet{ErdSpe:CMA1979}, we include a proof for the sake of clarity.

\begin{proposition}
[\protect{\citet{ErdSpe:CMA1979}}]
\label{pr:largest-connected-component}
Let $\probeq=0$, i.e., $\probless=1/2$,  let $\isoln_{\nplayers}$ be the set of isolated vertices in $\percol^{1/2}_{\nplayers}$, and recall that  $\nolarg^{1/2}_{\nplayers}$ is its fragment.
Then
\begin{equation}
\label{eq:P-M-Y}
    \sum_{\nplayers =1}^{\infty} \Prob\left(\nolarg^{1/2}_{\nplayers} \neq \isoln_{\nplayers}\right) <\infty.
\end{equation}
Moreover,
\begin{equation}
\label{eq:is}
\lim_{\nplayers \to \infty} \Prob\parens*{\card{\isoln_{\nplayers}} = \run} = \frac{\expo^{-1}}{\run!},\qquad \run \in \N.
\end{equation}
\end{proposition}

\begin{proof}
Choose $\delta$ as in \cref{le:McD2}. If  $\{\nolarg^{1/2}_{\nplayers} \neq \isoln_{\nplayers}\}$ then  there exists $\actprof$ such that $\card{\ball_{\ceil{\delta \nplayers}}(\actprof)\setminus {\larg}_{\nplayers}^{1/2,\actprofzero}} > 1$. 
Combining  \cref{le:McD2} (case $\probless =1/2$) with \cref{pr:percolation} we have
\begin{equation}
\label{eq:sumPcard}
\sum_{\nplayers} \Prob\parens*{\exists\; \actprofalt \colon \card{\ball_{\ceil{\delta \nplayers}}(\actprofalt)\setminus {\larg}_{\nplayers}^{1/2,\actprofzero}} > 1}   < \infty.
\end{equation}

This proves \cref{eq:P-M-Y}. To prove \cref{eq:is}  we reason as follows. 
The \acp{PNE} are the vertices  of $\orient{\hyperc}_{\nplayers}$ that are incident only to incoming edges.  
When $\probeq=0$, the number of \acp{PNE} has the same distribution as the  number of  vertices in $\orient{\hyperc}_{\nplayers}$ which are incident only to outgoing edges.  
Call $\orient{\isoln}_{\nplayers}$ the  set  of such vertices. 
It is well-known that, when $\probeq=0$, the number of \acp{PNE} converges in distribution to  a $\Poisson(1)$ (see \citet{ArrGolGor:AP1989,RinSca:GEB2000}). 
Hence, 
\begin{equation}
\label{eq:lim-P-card-Y}
\lim_{\nplayers \to \infty} \Prob\parens*{\card{\orient{\isoln}_{\nplayers}} = \run} = \frac{\expo^{-1}}{\run!}, \quad \run \in \N.
\end{equation}
We will prove that $\orient{\isoln}_{\nplayers} = {\isoln}_{\nplayers}$ for all  $\nplayers$ large enough.
If  $\actprofzero \notin \orient{\isoln}_{\nplayers}$ then all the vertices in  $\orient{\isoln}_{\nplayers}$ are not accessible from  $\actprofzero$.  
As 
\begin{equation}
\label{eq:sum1/2N}
\sum_{\nplayers=1}^\infty  \Prob(\actprofzero \in \orient{\isoln}_{\nplayers}) = \sum_{\nplayers=1}^\infty \frac{1}{2^{\nplayers}} <\infty,
\end{equation}
we can use the first Borel-Cantelli lemma  to conclude that for all $\nplayers$ large enough,  $\orient{\isoln}_{\nplayers} \subseteq  \orient\nolarg^{1/2,\actprofzero}_{\nplayers}$. 
    Using \cref{pr:percolation},  we have that $\orient{\isoln}_{\nplayers}\subseteq \vertices_{\nplayers} \setminus {\larg}_{\nplayers}^{1/2,\actprofzero}$. 
    Using \cref{pr:finally} and the first Borel-Cantelli lemma, we have that for   $\nplayers$ large enough, $\orient{\isoln}_{\nplayers} \subseteq \nolarg^{1/2}_{\nplayers}$. 
    Conversely, each isolated point $\actprof$ in $\percol^{1/2}_{\nplayers}$  whose neighbors are all in ${\larg}_{\nplayers}^{1/2,\actprofzero}$ satifies $\actprof \in \orient{\isoln}_{\nplayers}$.
    Hence, $\orient{\isoln}_{\nplayers} \neq \nolarg^{1/2}_{\nplayers}$ only if there exist two elements $\actprof, \actprofalt\in\nolarg^{1/2}_{\nplayers}$, such that $\hamm(\actprof, \actprofalt)=1$, i.e., $\actprof$ and $\actprofalt$ are neighbors in $\hyperc_{\nplayers}$. 
    Using \cref{eq:sumPcard}, we have
\begin{equation}
\label{eq:sum-P-Y-M}
\sum_{\nplayers=1}^\infty  \Prob\parens*{\orient{\isoln}_{\nplayers} \neq \nolarg^{1/2}_{\nplayers}}<\infty.
\end{equation}
Combining \cref{eq:sum-P-Y-M}  with \cref{eq:P-M-Y},  we have
\begin{equation}
\label{eq:boh-60}
\sum_{\nplayers=1}^\infty  \Prob\parens*{\orient{\isoln}_{\nplayers} \neq {\isoln}_{\nplayers}}<\infty.
\end{equation}
Using the first Borel-Cantelli lemma and \cref{eq:lim-P-card-Y} we prove \cref{eq:is}.
\end{proof}


\section{Proofs}
\label{se:proofs}


\subsection*{Proof of \cref{se:number}}

\begin{proof}
[Proof of \protect{\cref{th:SNE}}]
\noindent
\ref{it:th:SNE-1}
As we mentioned before, when $\probeq=0$, convergence of the number of \acp{PNE}  to a $\Poisson(1)$ was proved by \citet{ArrGolGor:AP1989,RinSca:GEB2000}. Moreover, since almost surely no two payoffs are equal, we have that each \ac{PNE} is also an \ac{SPNE}. 
	
\noindent
\ref{it:th:SNE-2}	Now we focus on the case $\probeq>0$ and prove that the number of \acp{SPNE} in $\rgame_{\nplayers}$ is zero  for all large $\nplayers$, $\Prob$-a.s..  
	Notice that $\probeq>0$ implies that $\probless<1/2$.
	
We have
\begin{equation}
\label{eq:card-SPNE}
\card{\sequilibria(\rgame_{\nplayers})}=\sum_{\actprof\in\vertices_{\nplayers}}\mathds{1}_{\actprof\in\sequilibria(\rgame_{\nplayers})},
\end{equation}
where $\mathds{1}_{\event}$ denotes the indicator function of the event $\event$.
Since $\Prob\parens*{\actprof\in\sequilibria(\rgame_{\nplayers})}=\probless^{\nplayers}$ for every $\actprof\in\vertices_{\nplayers}$, 
we have
\begin{equation}
\label{eq:E-card-SPNE}
\Expect\bracks{\card{\sequilibria(\rgame_{\nplayers})}}=(2\probless)^{\nplayers}.
\end{equation}
	
\noindent Markov's inequality implies
	\begin{equation}\label{eq:Markov-SNE}
	\Prob\parens*{\card{\sequilibria(\rgame_{\nplayers})}\ge 1} \le \Expect\bracks*{\card{ \sequilibria(\rgame_{\nplayers})}}= (2\probless)^{\nplayers}.
	\end{equation}
Since $2\probless<1$, the upper bound goes to zero geometrically fast.
Hence, 
\begin{equation}
\label{eq:sum-prob-SNE}
\sum_{\nplayers=1}^\infty \Prob (\card{{\sequilibria(\rgame_{\nplayers}}})\ge 1) < \infty.
\end{equation}
	Using the first Borel-Cantelli Lemma, we have that the event $\braces{\card{{\sequilibria(\rgame_{\nplayers}}})< 1}$ holds true for all  $\nplayers$ large enough. 
	As the cardinality is a non-negative integer, it must be zero for all large $\nplayers$.
\end{proof}

Given two probability measures $\Prob,\Probalt$ on $\N$, their total variation distance, $\dist_{\TV}$ is defined as
\begin{equation}
\label{eq:TV}
\dist_{\TV}(\Prob,\Probalt) \coloneqq \sup_{\event\subseteq\N}\abs{\Prob(\event)-\Probalt(\event)}.
\end{equation}
With an abuse of language, we will often speak of total variation distance of two random variables to indicate the total variation distance of their laws. 

\begin{remark}\label{re:ord}
Recall the definition of the Kolmogorov distance $\kolmo$, given in \cref{eq:Kolmogorov}. 
We have that 
\begin{equation}
\label{eq:Kolmogorov-TV} 
\kolmo(\Prob,\Probalt) \le \dist_{\TV}(\Prob,\Probalt).
\end{equation}
In fact, the first can be seen as the supremum over the collection of sets $\event = [0, \run]$ for $\run \in \N$ while the variational distance is a supremum over a richer collection of sets.
\end{remark}

To prove \cref{th:case-atoms} we will make use of the following lemma. 
\begin{lemma}[\protect{\citet[Theorem~4(c)]{McDScoWhi:arXiv2020}}]
\label{le:TV-Poisson}
Call $\comp_{\ncomp}$ the number of  connected components with exactly $\ncomp$  vertices in the percolation $\percol^{\probless}_{\nplayers}$ and let $\ecomp_{\ncomp} \coloneqq \Expect\bracks{\comp_{\ncomp}}$.
Then for $\probless \in (0, 1/2)$, we have
\begin{equation}
\label{eq:TV-Poisson}
\dist_{\TV}\parens*{\comp_{\ncomp},\Poisson(\ecomp_{\ncomp})} = \bigoh(\nplayers^{\ncomp}(1-\probless)^{\ncomp\nplayers}).
\end{equation}
\end{lemma}

\begin{proof}[Proof of \cref{th:case-atoms}]
Consider the random partially oriented hypercube $\orient{\hyperc}_{\nplayers}^{\probless}$ and invert the orientation of each edge, keeping the unoriented edges unoriented.  
The new random partially oriented hypercube $\rorient{\hyperc}_{\nplayers}^{\probless}$ has the same law as the original one.
Moreover, a vertex is a \ac{PNE} in $\orient{\hyperc}_{\nplayers}^{\probless}$ if and only if it is an isolated point in $\rorient{\hyperc}_{\nplayers}^{\probless}$. 

Consider the percolation associated to $\rorient{\hyperc}_{\nplayers}^{\probless}$ (see \cref{pr:percolation}). In virtue of our reasoning above, it is enough to establish a limit theorem for the number of isolated vertices of this percolation.
Resorting to \cref{le:TV-Poisson} with $\ncomp=1$, we obtain that there exists a constant $\consta_{1,\probeq}$, which depends on $\probeq$, but not on $\nplayers$, such that
\begin{equation}
\label{eq:TV-PNE-Poisson}
\dist_{\TV}\parens*{\card{\equilibria(\rgame_{\nplayers})},\Poisson(2^{\nplayers}\parens{1-\probless}^{\nplayers})}
\le \consta_{1,\probeq}(\nplayers(1-\probless)^{\nplayers}),
\end{equation}
which implies 
(see \cref{re:ord})
\begin{equation}
\label{eq:kolmoPNE-Poisson}
\kolmo\parens*{\card{\equilibria(\rgame_{\nplayers})},\Poisson(2^{\nplayers}\parens{1-\probless}^{\nplayers})}
\le \consta_{1,\probeq}(\nplayers(1-\probless)^{\nplayers}).
\end{equation}
Notice that $\Poisson(2^{\nplayers}\parens{1-\probless}^{\nplayers})$ can be expressed as the sum of $\floor{2^{\nplayers}\parens{1-\probless}^{\nplayers}}$ independent Poisson with parameter one, plus a small remainder.
In this context, we can use the Berry-Esseen theorem \citep[see, e.g.,][Theorem~3.4.17]{Dur:CUP2019}, which implies that
\begin{equation}
\label{eq:Poisson-Normal}
\kolmo\parens*{\Poisson(2^{\nplayers}\parens{1-\probless}^{\nplayers}), \Norm(2^{\nplayers}\parens{1-\probless}^{\nplayers},2^{\nplayers}\parens{1-\probless}^{\nplayers})} \le \frac{\consta_{2,\probeq}}{2^{\nplayers/2}\parens{1-\probless}^{\nplayers/2}}.
\end{equation}
For simplicity, we make the substitution $2(1-\probless) = 1+\probeq$. 
Using the triangular inequality we obtain
\begin{equation}
\label{eq:triangular-ineq}
\begin{split}
&\kolmo\parens*{\card{\equilibria(\rgame_{\nplayers})}, \Norm((1+\probeq)^{\nplayers},(1+\probeq)^{\nplayers})}\\
&\qquad\le \kolmo\parens*{\card{\equilibria(\rgame_{\nplayers})},\Poisson(\parens{1+\probeq}^{\nplayers})}\\ &\qquad\qquad+ \kolmo\parens*{\Poisson(\parens{1+\probeq}^{\nplayers}), \Norm(\parens{1+\probeq}^{\nplayers},\parens{1+\probeq}^{\nplayers})}\\
&\qquad\le 	\consta_{1,\probeq}\parens*{\nplayers\parens*{\frac{1+\probeq}{2}}^{\nplayers}}
+\frac{\consta_{2,\probeq}}{\parens{1+\probeq}^{\nplayers/2}}\\ 
&\qquad\le \nplayers\consta_{\probeq}\max\parens*{\frac{\parens{1+\probeq}}{2},\frac{1}{\parens{1+\probeq}^{1/2}}}^{\nplayers}.
\end{split}
\end{equation}
\end{proof}

\begin{proof}[Proof of \cref{th:geo}]
We first prove that 
\begin{equation}\label{eq:liminf}
\Prob\left(\liminf_{\nplayers\to\infty} \frac{\card{\equilibria(\rgame_{\nplayers})}}{(1+\probeq)^{\nplayers}}\ge 1\right) = 1.
\end{equation}
Notice that \cref{th:case-atoms} implies that the number of \ac{PNE} grows geometrically when $\probeq>0$. 
More precisely, we can show that,  for any $\varepsilon>0$ small enough,  \cref{eq:atoms-wvhp-2} implies that  
\begin{equation}\label{eq:ineq-N-equil1}
\sum_{\nplayers=1}^\infty \Prob\parens*{\card{\equilibria(\rgame_{\nplayers})} < (1+\probeq)^{\nplayers} -(1 + \varepsilon)^{\nplayers} (1+\probeq)^{\nplayers/2}} < \infty.
\end{equation}
Fix $\varepsilon >0$ such that 
\begin{equation*}
(1 + \varepsilon) < (1+\probeq)^{1/2}.
\end{equation*}
With this choice of $\varepsilon$, the right-hand side of \cref{eq:ineq-N-equil1} is $(1+\probeq)^{\nplayers}  (1 + o(1))$. Hence, by establishing \cref{eq:ineq-N-equil1} we would prove that the number of \ac{PNE} grows like $(1+\probeq)^\nplayers$.
 Recall the standard inequality   $ \normCDF(-x) \le \phi(x)/x $, valid for $x>0$, where $\phi$  is the density of a standard normal. We then have that 
\begin{equation*}
\normCDF(-(1+ \varepsilon)^{\nplayers})  \le \consta_{3} (1 + \varepsilon)^{-\nplayers}
\end{equation*}
for some constant $\consta_{3} >0$. Hence,
\begin{equation}
\label{eq:sum-card-PNE-ineq}
\begin{split}
&\quad\sum_{\nplayers=1}^\infty \Prob\parens*{\card{\equilibria(\rgame_{\nplayers})} 
\le (1+\probeq)^{\nplayers} -(1 + \varepsilon)^{\nplayers} (1+\probeq)^{\nplayers/2}}\\  
&=\sum_{\nplayers=1}^\infty \bracks*{\Prob\parens*{\card{\equilibria(\rgame_{\nplayers})}
\le (1+\probeq)^{\nplayers} -(1 + \varepsilon)^{\nplayers} (1+\probeq)^{\nplayers/2}}
 - \normCDF\parens*{-(1+ \varepsilon)^{\nplayers}}}\\
&\qquad+ \sum_{\nplayers=1}^\infty\normCDF(-(1+ \varepsilon)^{\nplayers})\\
&\le \sum_{\nplayers=1}^\infty \nplayers\consta_{\probeq}\max\parens*{\frac{\parens{1+\probeq}}{2},\frac{1}{\parens{1+\probeq}^{1/2}}}^{\nplayers}  +  \sum_{\nplayers=1}^\infty  \frac {\consta_{3}}{ (1 + \varepsilon)^{\nplayers}} < \infty.
\end{split}
\end{equation}

Similarly, we can prove the following 
   \begin{equation}\label{eq:ineq-N-equil1-2}
\sum_{\nplayers=1}^\infty \Prob\parens*{\card{\equilibria(\rgame_{\nplayers})} > (1+\probeq)^{\nplayers} +(1 + \varepsilon)^{\nplayers} (1+\probeq)^{\nplayers/2}} < \infty.
\end{equation}
This is achieved by repeating the argument above with the choice of $x_{\nplayers} = (1+ \varepsilon)^{\nplayers}$ and using the symmetry of normal distribution, i.e., $1- \normCDF(x) = \normCDF(-x)$.
In turn, \cref{eq:ineq-N-equil1-2} implies that 
\begin{equation}\label{eq:limsup}
\Prob\left(\limsup_{\nplayers \to \infty} \frac{\card{\equilibria(\rgame_{\nplayers})}}{(1+\probeq)^{\nplayers}}\le 1\right) = 1,
\end{equation}
which, together with \cref{eq:liminf}, ends the proof.
\end{proof}


\subsection*{Proofs of \cref{se:BRD}}

\begin{proof}
[Proof of \cref{th:accessibility}]
\noindent
\ref{it:th:accessibility-1}
First we deal with the case $\probeq =0$, which corresponds to $\probless = 1/2$, i.e., $\floorm_{\probless}=1$. 
In this case, there are no unoriented edges in $\orient{\hyperc}_{\nplayers}^{\probless}$.  
If a \ac{PNE} belongs to $ \orient\nolarg^{1/2,\actprofzero}_{\nplayers}$, then all its neighbors must also belong to $\orient\nolarg^{1/2,\actprofzero}_{\nplayers}$.   
In fact,  if one of these neighbors were accessible from $\actprofzero$, then also the \ac{PNE} would  be accessible  from $\actprofzero$, i.e., it would belong to $\orient\larg^{1/2,\actprofzero}_{\nplayers}$.
This means that there exists a ball of radius $1$ with $\nplayers+1$ vertices in $ \orient\nolarg^{1/2,\actprofzero}_{\nplayers}$. 
Hence, by \cref{le:McD2} (case $\probless =1/2$), we have
\begin{equation}
\label{eq:P-NE-L}
\begin{split}
\sum_{\nplayers =1}^{\infty}&\Prob\parens*{\equilibria(\rgame_{\nplayers})
\not\subseteq\orient{\larg}_{\nplayers}^{1/2,\actprofzero}}\le \sum_{\nplayers =1}^{\infty} \Prob\parens*{\exists\; \actprofalt \colon \card{\ball_{\ceil{\delta \nplayers}}(\actprofalt)\setminus\orient{\larg}_{\nplayers}^{1/2,\actprofzero}} > \nplayers} <\infty.
\end{split}
\end{equation}

 We now consider the case $0<\probeq < 1/2$.
	Let $\incid_{\nplayers}$ be the number of vertices that are incident to at least $\nplayers-\floorm_{\probless}$ unoriented edges.
	Markov's inequality yields 
\begin{equation}
\label{eq:ineq-incident}
\Prob\parens{\incid_{\nplayers}\ge 1}
	\le 2^{\nplayers}\sum_{\run=\nplayers-\floorm_{\probless}}^{\nplayers}\binom{\nplayers}{\run} \probeq^{\run}(1-\probeq)^{\nplayers - \run}
	\le \consta_{4,\probeq} \nplayers^{\floorm_\probless} 2^{\nplayers} \probeq^{\nplayers-\floorm_{\probless}},
\end{equation}
for some constant $\consta_{4,\probeq}$.

Hence, as $\probeq<1/2$,
\begin{equation}
	\sum_{\nplayers=1}^\infty \Prob\parens{\incid_{\nplayers}\ge 1}<\infty.
\end{equation}
Suppose that $\actprof$ is  a \ac{PNE} and belongs  to  $\orient\nolarg^{\beta,\actprofzero}_{\nplayers}$.  
Then, either:
\begin{enumerate}[(i)]
\item
\label{it:fragment-1}
the vertex $\actprof$ is adjacent to more than $\floorm_{\probless}$ oriented edges, which implies that the size of the  connected component of $ \orient\nolarg^{\beta,\actprofzero}_{\nplayers}$ containing $\actprof$ is larger than $\floorm_{\probless}$, or
\item	    
\label{it:fragment-2}
the vertex $\actprof$ is adjacent to at most $\floorm_{\probless}$ oriented edges.
\end{enumerate}
	
	When \ref{it:fragment-2} holds, we necessarily have that $\incid_{\nplayers}\ge 1$. Fix  $\delta>0$ as in \cref{le:McD2}.
	Hence,
\begin{multline*}
\sum_{\nplayers =1}^{\infty}\Prob\parens*{\equilibria(\rgame_{\nplayers})
\not\subseteq\orient{\larg}_{\nplayers}^{\probless,\actprofzero}}\\
\le \sum_{\nplayers=1}^{\infty} \Prob\parens*{\exists\; \actprofalt \colon \card{\ball_{\ceil{\delta \nplayers}}(\actprofalt) \setminus\orient{\larg}_{\nplayers}^{\probless,\actprofzero}} > \floorm_{\probless}}
+ \sum_{\nplayers=1}^\infty \Prob\parens{\incid_{\nplayers}\ge 1 }<\infty.
\end{multline*}

\noindent
\ref{it:th:accessibility-2}
	For this case, we  introduce a different percolation $\fairperc{\percol}_{\nplayers}^{\probless}$ on $\hyperc_{\nplayers}$ which is defined below. 
	This percolation is related to $\percol^{\probless}_{\nplayers}$, as defined in \cref{pr:percolation}. 
	For any pair of vertices $\actprofr,\actprofalt\in\vertices_{\nplayers}$, we declare the edge $\bracks{\actprofr,\actprofalt}$ open in $\fairperc{\percol}_{\nplayers}^{\probless}$ if	${\orientedges{\actprofr}{\actprofalt} \cup \orientedges{\actprofalt}{\actprofr}}$
{holds true, that is,} the edge connecting the two profiles $\actprofr$ and $\actprofalt$ is oriented in $\orient{\hyperc}_{\nplayers}^{\probless}$. 
	Otherwise the edge $\bracks{\actprofr,\actprofalt}$ is declared closed in $\fairperc{\percol}_{\nplayers}^{\probless}$.
	Since $\probeq =1/2$, the parameter of the percolation $\fairperc{\percol}_{\nplayers}^{\probless}$ is also  $1/2$ and we are in the framework studied in \citet{ErdSpe:CMA1979}, so we can apply \cref{pr:largest-connected-component}.

	Call $\fairperc{\larg}_{\nplayers}^{\probless,\actprofzero}$ the  connected component of   $\fairperc{\percol}_{\nplayers}^{\probless}$ that contains $\actprofzero$.
	Any isolated vertex in $\fairperc{\percol}_{\nplayers}^{\probless}$ is a \ac{PNE} in $\orient{\hyperc}_{\nplayers}^{\probless}$, as it is incident only to non-oriented edges, which in turn implies that each player has no incentive to deviate. 
	Using \cref{pr:largest-connected-component}, we have that the number of \acp{PNE} outside $\fairperc{\larg}_{\nplayers}^{\probless,\actprofzero}$ is asymptotically a $\Poisson(1)$ random variable.  
Notice that $\larg_{\nplayers}^{\probless,\actprofzero}\subseteq \fairperc{\larg}_{\nplayers}^{\probless,\actprofzero}$. 
This is because any edge that is open in $\percol^{\probless}_{\nplayers}$ is also open in $\fairperc{\percol}_{\nplayers}^{\probless}$. 
Using \cref{pr:percolation}, we have that $\orient{\larg}_{\nplayers}^{\probless,\actprofzero}\subseteq \fairperc{\larg}_{\nplayers}^{\probless,\actprofzero}$.
	Hence, the number of \acp{PNE} outside $\orient{\larg}_{\nplayers}^{\probless,\actprofzero}$ is stochastically larger than a $\Poisson(1)$ random variable.

	\noindent\ref{it:th:accessibility-3}
	For  $\actprofalt\in\vertices_{\nplayers}$, define $\equi{\actprofalt}$ to be the  event that vertex $\actprofalt$ is incident only to unoriented edges  in $\orient{\hyperc}_{\nplayers}^{\probless}$, and 
	\begin{equation}\label{eq:sum-isol}
	\equii{\nplayers}=\sum_{\actprofalt\in\vertices_{\nplayers}}\mathds{1}_{\equi{\actprofalt}}.
	\end{equation}
As we showed  after \cref{de:trap}, $\hyperc_{\nplayers}$ can be decomposed as
	\begin{equation}\label{eq:bipartite}
	\vertices_{\nplayers}=\vertices_{\nplayers}^{\even} \disjcup \vertices_{\nplayers}^{\odd},
	\end{equation}
	where $\vertices_{\nplayers}^{\even}$ is the set of vertices for which the sum of coordinates is even and $\vertices_{\nplayers}^{\odd}$ is the set of vertices for which the sum of coordinates is odd. 
	Edges connect only vertices from different components, so no pair of vertices in $\vertices_{\nplayers}^{\even}$ (or in $\vertices_{\nplayers}^{\odd}$) can be neighbors.

	Obviously $\card{\vertices_{\nplayers}^{\even}}=\card{\vertices_{\nplayers}^{\odd}}=2^{\nplayers-1}$.
	Our  first goal is to prove the following result.
\begin{lemma}
\label{le:independent-V-even}
The class $\braces{\equi{\actprofalt}:\actprofalt\in\vertices_{\nplayers}^{\even}}$  is a collection of independent events.
\end{lemma}

\begin{proof}
    	The event $\equi{\actprof}$ depends only on the payoffs at $\actprof$ and at each of its neighbors.
    	It is enough to prove that, for every subset $\subs\subseteq\vertices_{\nplayers}^{\even}$, we have
    	\begin{equation}
    	\label{eq:indepI}
    	\Prob\parens*{\bigcap_{\actprof\in\subs}\equi{\actprof}}=\prod_{\actprof\in\subs}\Prob\parens*{\equi{\actprof}}.
    	\end{equation}
    	Fix $\subs$ and $\actprofalt\in\subs$ and define $\subs_{-\actprofalt}\coloneqq\subs\setminus\braces{\actprofalt}$.
    	We need to prove that
    	\begin{equation}\label{eq:indepIalt}
    	\Prob\parens*{\bigcap_{\actprof\in\subs}\equi{\actprof}}
    	=\Prob\parens*{\equi{\actprofalt}\bigcap_{\actprof\in\subs_{-\actprofalt}}\equi{\actprof}}
    	=\Prob\parens*{\equi{\actprofalt}}\Prob\parens*{\bigcap_{\actprof\in\subs_{-\actprofalt}}\equi{\actprof}}.
    	\end{equation}	
    	The set of profiles in $\subs_{-\actprofalt}$ that share a neighbor with $\actprofalt$ has cardinality at most $\binom{\nplayers}{2}$.
    	If this set is empty, then \cref{eq:indepIalt} trivially holds.
    	Otherwise, for $\play\in\players$, let $\actprof^{\play\playalt}\in\subs_{-\actprofalt}$ and $\actprofw^{\play}$ be such that  $\actprofw^{\play}\sim_{\play}\actprofalt$ and  $\actprof^{\play\playalt}\sim_{\playalt}\actprofw^{\play}$, with $\play\neq\playalt$.
    	If, for some $\play$, the event $\equi{\actprof^{\play\playalt}}$ is true, then $\rand_{\playalt}^{\actprof^{\play\playalt}}=\rand_{\playalt}^{\actprofw^{\play}}$, and this event  is independent of $\rand_{\play}^{\actprofw^{\play}}$.
    	Therefore the class of events $\braces{\equi{\actprof^{\play\playalt}}}_{\play\in\players}$ is independent of the class of random variables $\braces{\rand_{\play}^{\actprofw^{\play}}}_{\play\in\players}$.
    	Since the event $\equi{\actprofalt}$ depends only on $\braces{\rand_{\play}^{\actprofw^{\play}}}_{\play\in\players}$ and $\rand_{\play}^{\actprofalt}$, we have that $\equi{\actprofalt}$ is independent of $\braces{\equi{\actprof^{\play\playalt}}}_{\play\in\players}$.
    	Moreover,  $\equi{\actprofalt}$ is independent of $\equi{\actprof}$ for all $\actprof\in\subs_{-\actprofalt}$. This ends the proof of  \cref{le:independent-V-even}.
\end{proof}
	
	As $\equii{\nplayers} \geq \sum_{\actprofalt\in\vertices_{\nplayers}^{\even}} \mathds{1}_{\equi{\actprofalt}}$, \cref{le:independent-V-even} implies that $\equii{\nplayers}$ is stochastically larger than a $\Binomial\parens*{2^{\nplayers-1},\probeq^{\nplayers}}$.
Each vertex $\actprofalt$ that is incident only to unoriented edges has the following properties:
\begin{itemize}
	\item
	it is a \ac{PNE}; and
		
	\item it lies in $\orient{\nolarg}_{\nplayers}^{\probless,\actprofzero}$,  unless $\actprofalt = \actprofzero$.
\end{itemize}
Hence, we have that  for any fixed $\consta >0$, 
	\begin{equation*}
    	\lim_{\nplayers\to\infty}\Prob\parens*{\card{\equilibria(\rgame_{\nplayers})\cap\orient{\nolarg}_{\nplayers}^{\probless,\actprofzero}}>\consta}=1.  \qedhere
	\end{equation*}	
\end{proof}

Recall  the definition of trap, given in \cref{de:trap}.
Notice that a trap  does not contain any \ac{PNE}. This is  because it is strongly connected whereas no  vertex is accessible from a \ac{PNE}.

\begin{proof}[Proof of \cref{th:BRD}]
The process $\brn$ does not converge to a \ac{PNE} if and only if it visits a trap.
Hence, it is enough to prove the  following stronger result about the existence of  traps 
\begin{equation}
\label{eq:trap-less-infty}
\sum_{\nplayers=1}^\infty \Prob(\exists\text{  a trap in }\orient{\hyperc}_{\nplayers}^{\probless}) 
< \infty.
\end{equation}
To prove \cref{eq:trap-less-infty}, we  first  study some properties of traps. 
Fix a trap $\orient{\trap}$. 
Each edge connecting $\orient{\trap}$ to its boundary $\boundary\orient{\trap}$ is either unoriented or points towards $\orient{\trap}$ (see the left of  \cref{fi:reverse-trap}).

	\begin{figure}[ht]
	    \centering
	    \begin{tikzpicture}[scale=0.5]
	        \node (0) at (0,0) {};
        	\node (1) at (4,0) {};
        	\node (2) at (0,4) {};
        	\node (3) at (4,4) {};
        	\node (4) at (-1.4, -1.4) {};
            \node (5) at (5.4, -1.4) {};
            \node (6) at (-1.4, 5.4) {};
            \node (7) at (5.4, 5.4) {};
            
            \node (8) at (12,0) {};
        	\node (9) at (16,0) {};
        	\node (10) at (12,4) {};
        	\node (11) at (16,4) {};
        	\node (12) at (10.6, -1.4) {};
            \node (13) at (17.4, -1.4) {};
            \node (14) at (10.6, 5.4) {};
            \node (15) at (17.4, 5.4) {};
            
        	
        	\filldraw (0,0) circle (3pt) [color=black];
        	\filldraw (4,0) circle (3pt) [color=black];
        	\filldraw (0,4) circle (3pt) [color=black];
        	\filldraw (4,4) circle (3pt) [color=black];
        	\filldraw (12,0) circle (3pt) [color=black];
        	\filldraw (16,0) circle (3pt) [color=black];
        	\filldraw (12,4) circle (3pt) [color=black];
        	\filldraw (16,4) circle (3pt) [color=black];
        	
        	\draw[thick] (5) -- (1);
        	\draw[thick] (6) -- (2);
        	\draw[thick] (13) -- (9);
        	\draw[thick] (14) -- (10);
        	
        	\begin{scope}[thick,decoration={markings, mark=at position 0.4 with {\arrow{latex}}, mark=at position 0.73 with {\arrow{latex}}}] 
        	    \draw[postaction={decorate}] (0) -- (1);
            	\draw[postaction={decorate}] (1) -- (3);
            	\draw[postaction={decorate}] (3) -- (2);
            	\draw[postaction={decorate}] (2) -- (0);
            	\draw[postaction={decorate}] (9) -- (8);
            	\draw[postaction={decorate}] (11) -- (9);
            	\draw[postaction={decorate}] (10) -- (11);
            	\draw[postaction={decorate}] (8) -- (10);
	        \end{scope}
	        
	        \begin{scope}[thick,decoration={markings, mark=at position 0.7 with {\arrow{latex}}}]
	            \draw[postaction={decorate}] (4) -- (0);
        	    \draw[postaction={decorate}] (7) -- (3);
        	    \draw[postaction={decorate}] (8) -- (12);
        	    \draw[postaction={decorate}] (11) -- (15);
	        \end{scope}
	    \end{tikzpicture}
	    \caption{A trap of size 4 (left) and the corresponding subgraph after reversing the orientation on the edges (right). Notice that, after reversing the edges, there is no way for the process to enter the cycle from the outside. \label{fi:reverse-trap}}
	\end{figure}
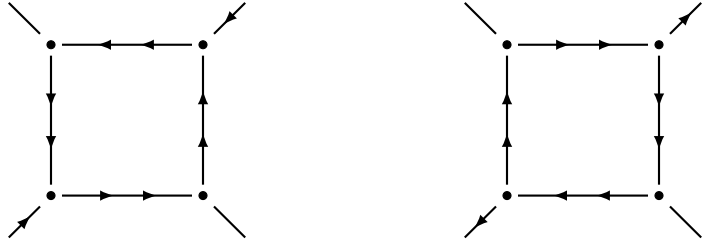
	
	Invert the orientation of the partially oriented hypercube $\orient{\hyperc}_{\nplayers}^{\probless}$, while leaving the unoriented edges unchanged, to obtain $\rorient{\hyperc}_{\nplayers}^{\probless}$.  
	This is illustrated in \cref{fi:reverse-trap}.  The  random partially oriented graphs $\orient{\hyperc}_{\nplayers}^{\probless}$ and $\rorient{\hyperc}_{\nplayers}^{\probless}$ share the same distribution. 
Consider the partition 
\begin{equation}
\label{eq:partition01}
\vertices_{\nplayers}=\rorient{\larg}_{\nplayers}^{\probless,\actprofzero}\disjcup\rorient{\nolarg}_{\nplayers}^{\probless,\actprofzero}
\end{equation}
in such a way that $\rorient{\larg}_{\nplayers}^{\probless,\actprofzero}$ is the set that contains $\actprofzero$ as well as all vertices $\actprofalt$ that are accessible from $\actprofzero$ in the oriented graph $\rorient{\hyperc}_{\nplayers}^{\probless}$. Conversely, all vertices that are \emph{not} accessible from $\actprofzero$ are contained in $\rorient{\nolarg}_{\nplayers}^{\probless,\actprofzero}$ in the oriented graph $\rorient{\hyperc}_{\nplayers}^{\probless}$. Let $\rorient{\trap}$  be the subgraph in $\rorient{\hyperc}_{\nplayers}^{\probless}$ corresponding to the trap $\orient{\trap}\subseteq\orient{\hyperc}_{\nplayers}^{\probless}$. 
More precisely, $\rorient{\trap}$  and $\orient{\trap}$ share the same vertex set and edges,  either unoriented or with opposite orientation.  
We call  $\rorient{\trap}$ a  reversed trap.

Note that all edges in $\boundary\rorient{\trap}$ are either unoriented or oriented away from $\rorient{\trap}$.   
Hence, for all large $\nplayers$, either:
\begin{enumerate}[(i)]
\item 
\label{it:V-trap-1}
$\actprofzero \notin \vertices({\rorient{\trap}})$, which implies that $\vertices({\rorient{\trap}})$ is contained in $\rorient{\nolarg}_{\nplayers}^{\probless,\actprofzero}$, or

\item
\label{it:V-trap-2}
$\actprofzero \in \vertices({\rorient{\trap}})$,  which implies (using \cref{eq:trap}) that $\vertices({\rorient{\trap}})= \orient{\larg}_{\nplayers}^{\probless,\actprofzero}$.
\end{enumerate} 

We consider the two cases separately, and focus on case \ref{it:V-trap-1} first. We  prove that  
\begin{equation}\label{eq:ex-trap}
\sum_{\nplayers=1}^\infty \Prob(\exists \text{ a reversed trap }\rorient{\trap}\text{ such that } \vertices({\rorient{\trap}}) \subseteq \rorient{\nolarg}_{\nplayers}^{\probless,\actprofzero})
< \infty.
\end{equation}
Recall that any (reversed) trap has at least four vertices and notice that $\rorient{\nolarg}_{\nplayers}^{\probless,\actprofzero}$ and $\orient{\nolarg}_{\nplayers}^{\probless,\actprofzero}$ are equally distributed.
By \cref{le:McD2}, there exists $\delta>0$ such that \cref{eq:sum-McD} holds.
By \cref{eq:convergence}, we have
\begin{equation*}
\floorm_{\probless}=\floor*{\frac{1}{-\log_2(1-\probless)}}= \floor*{- \frac 1{\log_{2}{(1/2 + \probeq/2)}}} \le 3.
\end{equation*}
Hence, 
\begin{equation}
\label{eq:sum-less-than-4}
\sum_{\nplayers=1}^\infty \Prob(\exists\text{ a connected component of size at least  $4$ whose vertices are in  $\rorient{\nolarg}_{\nplayers}^{\probless,\actprofzero}$ })<\infty,    
\end{equation}
which implies  \cref{eq:ex-trap}.
 
We now focus on case \ref{it:V-trap-2} and prove that 
\begin{equation}
\label{eq:P-lcc-B'}
\sum_{\nplayers=1}^\infty \Prob(\exists \text{ a  trap }\orient{\trap}\text{ such that } \vertices({\orient{\trap}}) = \orient{\larg}_{\nplayers}^{\probless,\actprofzero} ) < \infty.
\end{equation}
As a  trap cannot contain any \ac{PNE},  to show that \cref{eq:P-lcc-B'} holds it is enough to prove that 
\begin{equation}
\label{eq:B'-no-PNE}
\sum_{\nplayers=1}^\infty \Prob(\orient{\larg}_{\nplayers}^{\probless,\actprofzero}  \text{ contains no \ac{PNE}})<\infty.
\end{equation}
To this end, we introduce  a new  class of events.
For every strategy $\actprof$, define the event 
\begin{equation}
\label{eq:G-s}
\threeincid_{\actprof} \coloneqq \braces*{\actprof \text{ is incident to at most $2$ oriented edges in }  \orient{\hyperc}_{\nplayers}^{\probless}}.
\end{equation}
Moreover, set
\begin{equation}
\label{eq:H-s}
\begin{split}
\threeinter_{\actprof} &\coloneqq  \threeincid_{\actprof}^{c}\cap\braces*{\actprof  \in \orient{\nolarg}_{\nplayers}^{\probless,\actprofzero} }\cap \{\actprof \text{ is a \ac{PNE}}\}. 
\end{split}
\end{equation} 
Fix $\varepsilon< \sqrt{1+\probeq} -1 $ and set 
\begin{equation}
\label{eq:ell}
\thresh_{\nplayers}= (1+\probeq)^{\nplayers} -(1 + \varepsilon)^{\nplayers} (1+\probeq)^{\nplayers/2}.
\end{equation}
We have that
\begin{equation}\label{eq:L-no-PE}
\braces*{\orient{\larg}_{\nplayers}^{\probless,\actprofzero}  \text{ contains no \ac{PNE}}}
\subseteq
\braces*{\card{\equilibria(\rgame_{\nplayers})} < \thresh_{\nplayers}}
\cup \braces*{\sum_{\actprof  \in \vertices_{\nplayers}} \mathds{1}_{\threeincid_{\actprof}} \ge \thresh_{\nplayers}}  
\cup \braces*{\bigcup_{\actprof \in \vertices_{\nplayers}}\threeinter_{\actprof}}.
\end{equation}
To see \cref{eq:L-no-PE}, notice that  the event in which all the \acp{PNE} are in $\orient{\nolarg}_{\nplayers}^{\probless,\actprofzero}$ can be decomposed into two disjoint subevents:
\begin{enumerate}[(a)]
\item
\label{it:PNE-fragment-1} 
the number of \acp{PNE} is smaller than $\thresh_{\nplayers}$; and
\item 
\label{it:PNE-fragment-2} 
the number of  \acp{PNE} is at least $\thresh_{\nplayers}$. 
\end{enumerate}
Then, under case \ref{it:PNE-fragment-2}, at least one of the following must hold: 
\begin{enumerate}[(i)]
\item
\label{it:sub-number-PNE-1}
there exists a \ac{PNE}  in $\orient{\nolarg}_{\nplayers}^{\probless,\actprofzero}$ which is incident to at least three directed edges so $\bigcup_{\actprof \in \vertices_{\nplayers}}\threeinter_{\actprof}$ holds; or
\item
\label{it:sub-number-PNE-2} all the \acp{PNE} are incident to at most two directed edges, which implies
$\sum_{\actprof \in \vertices_{\nplayers}} \mathds{1}_{\threeincid_{\actprof}} \ge \thresh_{\nplayers}$. 
\end{enumerate}

Next, we show how to use \cref{eq:L-no-PE} to prove \cref{eq:B'-no-PNE}. 

\noindent
Bound for \ref{it:PNE-fragment-1}:
We recall \cref{eq:ineq-N-equil1}, which  can be rewritten as 
\begin{equation}
\label{eq:sum-P-card-PNE}
 \sum_{\nplayers=1}^\infty  \Prob\parens*{\card{\equilibria(\rgame_{\nplayers})} < \thresh_{\nplayers}} < \infty.
\end{equation}

\noindent
Bound for \ref{it:PNE-fragment-2}\ref{it:sub-number-PNE-1}: 
For any $\actprof \in \vertices_{\nplayers}$, we prove that 
\begin{equation}
\label{eq:Hs-subset} 
\threeinter_{\actprof} \subseteq \braces*{\exists\text{ a connected component of size at least }4 \text{ whose vertices are  in }\rorient{\nolarg}_{\nplayers}^{\probless,\actprofzero}}.
\end{equation}
In fact, under the event $\threeinter_{\actprof}$, the vertex $\actprof$ is in $\orient{\nolarg}_{\nplayers}^{\probless,\actprofzero}$ and is incident to at least three oriented edges. 
Owing to $\actprof$ being a \ac{PNE}, each of these edges points toward $\actprof$. 
Hence, there are (at least) three other vertices in the connected component where $\actprof$ belongs, proving \cref{eq:Hs-subset}.
As a consequence, we have
\begin{equation}
\label{eq:union-Hs-subset} 
\bigcup_{\actprof \in \vertices_{\nplayers}} \threeinter_{\actprof} \subseteq \braces*{\exists\text{ a connected component of size at least }4\text{ whose vertices are  in }\rorient{\nolarg}_{\nplayers}^{\probless,\actprofzero}}.
\end{equation}
Therefore,
\begin{equation}
\label{eq:spe2}
\begin{split}
&\sum_{\nplayers=1}^\infty \Prob\parens*{\bigcup_{\actprof \in \vertices_{\nplayers}} \threeinter_{\actprof}}\\
&\le 
\sum_{\nplayers=1}^\infty  \Prob\parens*{\exists\text{ a connected component of size at least  $4$ whose vertices are  in  $\rorient{\nolarg}_{\nplayers}^{\probless,\actprofzero}$ }}<\infty.
\end{split}
\end{equation}
\noindent
Bound for \ref{it:PNE-fragment-2}\ref{it:sub-number-PNE-2}:
We have that
\begin{equation}\label{eq:P-G-s}
\Prob(\threeincid_{\actprof}) = \sum_{\run=0}^{2} \binom{\nplayers}{\run}  \probeq^{\nplayers-\run} (1- \probeq)^{\run}  \le  \consta_{5,\probeq} \nplayers^{2}  \probeq^{\nplayers},
\end{equation}
for some constant $\consta_{5,\probeq}$, which depends only on $\probeq$.
Using Markov's inequality, we have
\begin{equation}
\label{eq:spe1}
\sum_{\nplayers=1}^\infty \Prob\parens*{\sum_{\actprof} \mathds{1}_{\threeincid_{\actprof}} \ge \thresh_{\nplayers}} \le 
\sum_{\nplayers=1}^\infty  \consta_{5,\probeq} \nplayers^2  \frac{(2\probeq)^{\nplayers}}{\thresh_{\nplayers}} <\infty.
\end{equation}
The convergence of the sum is a consequence of $1+\probeq > 2 \probeq$, and the fact that 
$\thresh_{\nplayers} = (1+\probeq)^{\nplayers}(1+ o(1))$.

\cref{eq:sum-P-card-PNE,eq:spe1,eq:spe2} show that the subevents in \cref{eq:L-no-PE} are summable, proving \cref{eq:B'-no-PNE}. Furthermore, \cref{eq:ex-trap,eq:B'-no-PNE} imply \cref{eq:trap-less-infty}.
\end{proof}


\section{Conclusions and open problems}
\label{se:conclusions}

Large random games have many  \acp{PNE}, as long as the probability of ties is nonzero. 
We identified  the limiting distribution of the number of \acp{PNE} and their position with respect to the starting point of a \ac{BRD}. 
More specifically, in \cref{th:case-atoms,th:geo} we have proved that the number of \acp{PNE} grows geometrically when $\probeq>0$.
In \cref{th:accessibility} we have described the set of \acp{PNE} and have proved that, when $\probeq$ is larger than a certain threshold, some of them are not accessible.
In \cref{th:BRD}, instead, we have established that, when $\probeq$ is small enough, the \ac{BRD} converges to a \ac{PNE}. 
The relevance of our approach is that it creates a link between different subjects. 

The next important question is the following. 
How long does it take for a \ac{BRD} to reach a \ac{PNE}? 
This is equivalent to studying the path-length of a non-backtracking random walk on the percolation cluster of the hypercube. 
\cref{fi:arbre3} shows the results of a simulation exploring this problem. 

\begin{figure}[ht]  
	\centering 
	\includegraphics[width=8cm]{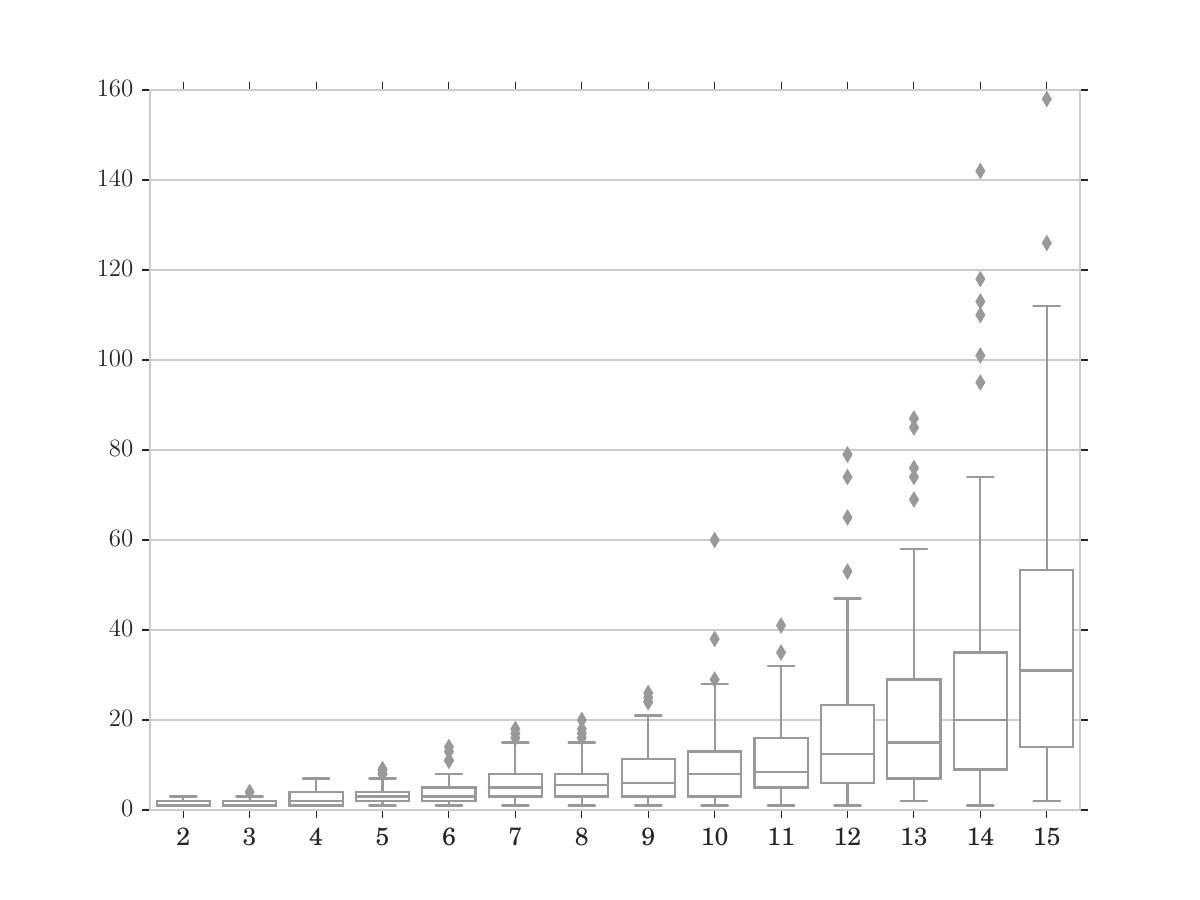}
	\caption{\label{fi:arbre3} Iterations needed for BRD to reach a \ac{PNE} for $\probeq=0.5$, with 100 trials per $\nplayers$.}
\end{figure}

Then, it is important to study the geometry of \acp{PNE} when more strategies are available, and when the payoffs are weakly dependent.


\subsection*{Acknowledgements}
	Marco Scarsini is a member of GNAMPA-INdAM and of the COST Action GAMENET.
	He gratefully acknowledges the support and hospitality of the Department of Mathematics at Monash University, where this research started. 
	His work is partially supported by the GNAMPA-INdAM Project 2019 ``Markov chains and games on networks'' and by the PRIN 2017 project ALGADIMAR.
	Andrea Collevecchio's  work is partially supported by ARC grant  DP180100613 and Australian Research Council Centre of Excellence for Mathematical and Statistical Frontiers (ACEMS) CE140100049. 
	Ben Amiet and Ziwen Zhong  are supported by an Australian Government Research Training Program (RTP) Scholarship.
	The authors thank Alfie Mimun, Matteo Quattropani and Alex Scott for some very helpful comments, Mikhail Isaev for pointing out to one of us the paper by  \citet{McDScoWhi:arXiv2020}. 
	Moreover, the authors thank the Associate Editor and the referees for comments that  improved the paper and for mentioning the paper by \citet{Lin:arXiv2009}.

\bibliography{../bibtex/bibNEpercolation}

\begin{thebibliography}{}

\bibitem[Arratia et~al., 1989]{ArrGolGor:AP1989}
Arratia, R., Goldstein, L., and Gordon, L. (1989).
\newblock Two moments suffice for {P}oisson approximations: the {C}hen-{S}tein
  method.
\newblock {\em Ann. Probab.}, 17(1):9--25.

\bibitem[Bernheim, 1984]{Ber:E1984}
Bernheim, B.~D. (1984).
\newblock Rationalizable strategic behavior.
\newblock {\em Econometrica}, 52(4):1007--1028.

\bibitem[Blum and Mansour, 2007]{BluMan:AGT2007}
Blum, A. and Mansour, Y. (2007).
\newblock Learning, regret minimization, and equilibria.
\newblock In {\em Algorithmic game theory}, pages 79--101. Cambridge Univ.
  Press, Cambridge.

\bibitem[Blume, 1993]{Blu:GEB1993}
Blume, L.~E. (1993).
\newblock The statistical mechanics of strategic interaction.
\newblock {\em Games Econom. Behav.}, 5(3):387--424.

\bibitem[Bollob\'{a}s, 2001]{Bol:CUP2001}
Bollob\'{a}s, B. (2001).
\newblock {\em Random Graphs}.
\newblock Cambridge University Press, Cambridge, second edition.

\bibitem[Borel, 1921]{Bor:CRASP1921}
Borel, E. (1921).
\newblock La th\'eorie du jeu et les \'equations int\'egrales \`a noyau
  sym\'etrique.
\newblock {\em C. R. Acad. Sci. Paris}, 173:1304--1308.

\bibitem[Broadbent and Hammersley, 1957]{BroHam:PCPS1957}
Broadbent, S.~R. and Hammersley, J.~M. (1957).
\newblock Percolation processes. {I}. {C}rystals and mazes.
\newblock {\em Proc. Cambridge Philos. Soc.}, 53:629--641.

\bibitem[Candogan et~al., 2011]{CanMenOzdPar:MOR2011}
Candogan, O., Menache, I., Ozdaglar, A., and Parrilo, P.~A. (2011).
\newblock Flows and decompositions of games: harmonic and potential games.
\newblock {\em Math. Oper. Res.}, 36(3):474--503.

\bibitem[Christodoulou et~al., 2012]{ChrMirSid:TCS2012}
Christodoulou, G., Mirrokni, V.~S., and Sidiropoulos, A. (2012).
\newblock Convergence and approximation in potential games.
\newblock {\em Theoret. Comput. Sci.}, 438:13--27.

\bibitem[Cohen, 1998]{Coh:PNAS1998}
Cohen, J.~E. (1998).
\newblock Cooperation and self-interest: {P}areto-inefficiency of {N}ash
  equilibria in finite random games.
\newblock {\em Proc. Natl. Acad. Sci. USA}, 95(17):9724--9731.

\bibitem[Coucheney et~al., 2014]{CouDurGauTou:NetGCoop2014}
Coucheney, P., Durand, S., Gaujal, B., and Touati, C. (2014).
\newblock General revision protocols in best response algorithms for potential
  games.
\newblock In {\em Netwok Games, Control and OPtimization (NetGCoop)}, Trento,
  Italy. IEEE Explore.

\bibitem[Daskalakis et~al., 2011]{DasDimMos:AAP2011}
Daskalakis, C., Dimakis, A.~G., and Mossel, E. (2011).
\newblock Connectivity and equilibrium in random games.
\newblock {\em Ann. Appl. Probab.}, 21(3):987--1016.

\bibitem[Daskalakis et~al., 2009]{DasGolPap:SIAMJC2009}
Daskalakis, C., Goldberg, P.~W., and Papadimitriou, C.~H. (2009).
\newblock The complexity of computing a {N}ash equilibrium.
\newblock {\em SIAM J. Comput.}, 39(1):195--259.

\bibitem[Dresher, 1970]{Dre:JCT1970}
Dresher, M. (1970).
\newblock Probability of a pure equilibrium point in {$n$}-person games.
\newblock {\em J. Combinatorial Theory}, 8:134--145.

\bibitem[Durand et~al., 2019]{DurGarGau:PE2019}
Durand, S., Garin, F., and Gaujal, B. (2019).
\newblock Distributed best response dynamics with high playing rates in
  potential games.
\newblock {\em Performance Evaluation}, 129:40--59.

\bibitem[Durand and Gaujal, 2016]{DurGau:AGT2016}
Durand, S. and Gaujal, B. (2016).
\newblock Complexity and optimality of the best response algorithm in random
  potential games.
\newblock In {\em Algorithmic Game Theory}, volume 9928 of {\em Lecture Notes
  in Comput. Sci.}, pages 40--51. Springer, Berlin.

\bibitem[Durrett, 2019]{Dur:CUP2019}
Durrett, R. (2019).
\newblock {\em Probability---Theory and Examples}.
\newblock Cambridge University Press, Cambridge, fifth edition.

\bibitem[D\"{u}tting and Kesselheim, 2017]{DutKes:SODA2017}
D\"{u}tting, P. and Kesselheim, T. (2017).
\newblock Best-response dynamics in combinatorial auctions with item bidding.
\newblock In {\em Proceedings of the {T}wenty-{E}ighth {A}nnual {ACM}-{SIAM}
  {S}ymposium on {D}iscrete {A}lgorithms}, pages 521--533. SIAM, Philadelphia,
  PA.

\bibitem[Erd\H{o}s and Spencer, 1979]{ErdSpe:CMA1979}
Erd\H{o}s, P. and Spencer, J. (1979).
\newblock Evolution of the {$n$}-cube.
\newblock {\em Comput. Math. Appl.}, 5(1):33--39.

\bibitem[Fabrikant et~al., 2013]{FabJagSha:TCS2013}
Fabrikant, A., Jaggard, A.~D., and Schapira, M. (2013).
\newblock On the structure of weakly acyclic games.
\newblock {\em Theory Comput. Syst.}, 53(1):107--122.

\bibitem[Friedman and Mezzetti, 2001]{FriMez:JET2001}
Friedman, J.~W. and Mezzetti, C. (2001).
\newblock Learning in games by random sampling.
\newblock {\em J. Econom. Theory}, 98(1):55--84.

\bibitem[Galla and Farmer, 2013]{GalFar:PNAS2013}
Galla, T. and Farmer, J.~D. (2013).
\newblock Complex dynamics in learning complicated games.
\newblock {\em Proc. Natl. Acad. Sci. USA}, 110(4):1232--1236.

\bibitem[Goemans et~al., 2005]{GoeMirVet:FOCS2005}
Goemans, M., Mirrokni, V., and Vetta, A. (2005).
\newblock Sink equilibria and convergence.
\newblock In {\em 46th Annual IEEE Symposium on Foundations of Computer Science
  (FOCS'05)}, pages 142--151.

\bibitem[Goldberg et~al., 1968]{GolGolNew:JRNBSB1968}
Goldberg, K., Goldman, A.~J., and Newman, M. (1968).
\newblock The probability of an equilibrium point.
\newblock {\em J. Res. Nat. Bur. Standards Sect. B}, 72B:93--101.

\bibitem[Goldman, 1957]{Gol:AMM1957}
Goldman, A.~J. (1957).
\newblock The probability of a saddlepoint.
\newblock {\em Amer. Math. Monthly}, 64:729--730.

\bibitem[Grimmett, 1999]{Gri:Springer1999}
Grimmett, G. (1999).
\newblock {\em Percolation}.
\newblock Springer-Verlag, Berlin, second edition.

\bibitem[Linusson, 2009]{Lin:arXiv2009}
Linusson, S. (2009).
\newblock A note on correlations in randomly oriented graphs.
\newblock Technical report, arXiv:0905.2881.

\bibitem[McDiarmid et~al., 2020]{McDScoWhi:arXiv2020}
McDiarmid, C., Scott, A., and Withers, P. (2020).
\newblock The component structure of dense random subgraphs of the hypercube.
\newblock Technical report, arXiv:1806.06433.

\bibitem[Monderer and Shapley, 1996]{MonSha:GEB1996}
Monderer, D. and Shapley, L.~S. (1996).
\newblock Potential games.
\newblock {\em Games Econom. Behav.}, 14(1):124--143.

\bibitem[Nash, 1951]{Nash:AM1951}
Nash, J. (1951).
\newblock Non-cooperative games.
\newblock {\em Ann. of Math. (2)}, 54:286--295.

\bibitem[Nash, 1950]{Nash:PNAS1950}
Nash, Jr., J.~F. (1950).
\newblock Equilibrium points in {$n$}-person games.
\newblock {\em Proc. Nat. Acad. Sci. U. S. A.}, 36:48--49.

\bibitem[Osborne and Rubinstein, 1994]{OsbRub:MITPress1994}
Osborne, M.~J. and Rubinstein, A. (1994).
\newblock {\em A Course in Game Theory}.
\newblock MIT Press, Cambridge, MA.

\bibitem[Pangallo et~al., 2019]{PanHeiFar:SA2019}
Pangallo, M., Heinrich, T., and Doyne~Farmer, J. (2019).
\newblock Best reply structure and equilibrium convergence in generic games.
\newblock {\em Science Advances}, 5(2).

\bibitem[Pei and Takahashi, 2019]{PeiTak:GEB2019}
Pei, T. and Takahashi, S. (2019).
\newblock Rationalizable strategies in random games.
\newblock {\em Games Econom. Behav.}, 118:110--125.

\bibitem[Powers, 1990]{Pow:IJGT1990}
Powers, I.~Y. (1990).
\newblock Limiting distributions of the number of pure strategy {N}ash
  equilibria in {$N$}-person games.
\newblock {\em Internat. J. Game Theory}, 19(3):277--286.

\bibitem[Rai\v{c}, 2003]{Rai:P7YSM2003}
Rai\v{c}, M. (2003).
\newblock Normal approximation by {S}tein's method.
\newblock In Mrvar, A., editor, {\em Proceedings of the Seventh Young
  Statisticians Meeting}, pages 71--97.

\bibitem[Rinott and Scarsini, 2000]{RinSca:GEB2000}
Rinott, Y. and Scarsini, M. (2000).
\newblock On the number of pure strategy {N}ash equilibria in random games.
\newblock {\em Games Econom. Behav.}, 33(2):274--293.

\bibitem[Stanford, 1995]{Sta:GEB1995}
Stanford, W. (1995).
\newblock A note on the probability of {$k$} pure {N}ash equilibria in matrix
  games.
\newblock {\em Games Econom. Behav.}, 9(2):238--246.

\bibitem[Stanford, 1996]{Sta:MOR1996}
Stanford, W. (1996).
\newblock The limit distribution of pure strategy {N}ash equilibria in
  symmetric bimatrix games.
\newblock {\em Math. Oper. Res.}, 21(3):726--733.

\bibitem[Stanford, 1997]{Sta:MSS1997}
Stanford, W. (1997).
\newblock On the distribution of pure strategy equilibria in finite games with
  vector payoffs.
\newblock {\em Math. Social Sci.}, 33(2):115--127.

\bibitem[Stanford, 1999]{Sta:EL1999}
Stanford, W. (1999).
\newblock On the number of pure strategy {N}ash equilibria in finite common
  payoffs games.
\newblock {\em Econom. Lett.}, 62(1):29--34.

\bibitem[Takahashi, 2008]{Tak:GEB2008}
Takahashi, S. (2008).
\newblock The number of pure {N}ash equilibria in a random game with
  nondecreasing best responses.
\newblock {\em Games Econom. Behav.}, 63(1):328--340.

\bibitem[Takahashi and Yamamori, 2002]{TakYam:EB2002}
Takahashi, S. and Yamamori, T. (2002).
\newblock The pure {N}ash equilibrium property and the quasi-acyclic condition.
\newblock {\em Econ. Bull.}

\bibitem[Tardos and Vazirani, 2007]{TarVaz:AGT20017}
Tardos, E. and Vazirani, V.~V. (2007).
\newblock Basic solution concepts and computational issues.
\newblock In {\em Algorithmic game theory}, pages 3--28. Cambridge Univ. Press,
  Cambridge.

\bibitem[von Neumann, 1928]{vNe:MA1928}
von Neumann, J. (1928).
\newblock Zur {T}heorie der {G}esellschaftsspiele.
\newblock {\em Math. Ann.}, 100(1):295--320.

\bibitem[Watson, 1995]{Wat:CUP1995}
Watson, G.~N. (1995).
\newblock {\em A Treatise on the Theory of {B}essel Functions}.
\newblock Cambridge University Press, Cambridge.
\newblock Reprint of the second (1944) edition.

\bibitem[Young, 1993]{You:E1993}
Young, H.~P. (1993).
\newblock The evolution of conventions.
\newblock {\em Econometrica}, 61(1):57--84.

\end{thebibliography}
\bibliographystyle{apalike}

\section{List of symbols}
\begin{longtable}{p{.13\textwidth} p{.82\textwidth}}

$\ball_{\radius}(\actprof)$ & ball of radius $r$ centred at $\actprof$, introduced in \cref{le:McD2}\\
$\percol_{\run}$ & percolation process at time $\run$, defined in \cref{eq:Bk+1}\\
$\percol^{\probless}_{\nplayers}$ & edge percolation coupled with random game\\
$\fairperc{\percol}_{\nplayers}$ & percolation process, defined in proof of \cref{th:accessibility} \ref{it:th:accessibility-2}\\
$\brn$ & \ac{BRD} on $\hyperc_{\nplayers}$\\
$\comp_{\ncomp}$ & number of components of size $\ncomp$ in the percolation $\percol^{*}$ \\
$\card{\event}$ & cardinality of the set $\event$\\
$\profitable(\actprof)$ & set of profitable deviations from $\act$, define in \cref{eq:profitable}\\
$\edges_{\nplayers}$ &  edge set of $\hyperc_{\nplayers}$\\
$\orient{\edges}_{\nplayers}$ & edge set of $\orient{\hyperc}_{\nplayers}^{\probless}$\\
$\mathscr{F}$ & $\sigma$-algebra associated with the probability space, introduced after \cref{de:Nash-equilibrium}\\
$\payoff_{\play}$ & payoff function for player $i$, introduced in \cref{eq:game-def}\\
$\threeincid_{\actprof}$ & event defined in \cref{eq:G-s}\\
$\hamm$ &  Hamming distance on $\hyperc_{\nplayers}$, defined in \cref{eq:Hamming}\\
$\hyperc_{\nplayers}$ &  $\nplayers$-cube\\
$\orient{\hyperc}_{\nplayers}^{\probless}$ & partially oriented hypercube\\
$\rorient{\hyperc}_{\nplayers}^{\probless}$ & the graph obtained by reversing the oriented edges in $\orient{\hyperc}_{\nplayers}^{\probless}$, introduced in the proof of \cref{th:case-atoms}\\
$\subs$ & a subset of $\vertices_{\nplayers}^{\even}$, introduced in \cref{eq:indepI}\\
$\subs_{-\actprofalt}$ & equal to $\subs\setminus\braces{\actprofalt}$, defined in \cref{eq:indepIalt}\\
$\consta_\probeq$ & a constant dependent on $\probeq$, introduced in \cref{th:case-atoms}\\
$\thresh_{\nplayers}$ & defined in \cref{eq:ell}\\
$\orient{\larg}_{\nplayers}^{\probless,\actprofzero}$ & set of all vertices in $\orient{\hyperc}_{\nplayers}^{\probless}$ accessible from $\actprofzero$, introduced in \cref{eq:partition}\\
$\rorient{\larg}_{\nplayers}^{\probless,\actprofzero}$ & set of all vertices in $\rorient{\hyperc}_{\nplayers}^{\probless}$ accessible from $\actprofzero$, introduced in \cref{eq:partition01}\\
$\larg_{\nplayers}^{\probless}$ & giant component of $\percol^{\probless}_{\nplayers}$\\
$\larg_{\nplayers}^{\probless, \actprofzero}$ & connected component of $\percol^{\probless}_{\nplayers}$ containing $\actprofzero$\\
$\floorm_{\probless}$ & constant, introduced in \cref{le:McD2}\\
$\nolarg_{\nplayers}^{\probless}$ & the fragment of the percolation, introduced in the proof of \cref{pr:finally}\\
$\orient{\nolarg}_{\nplayers}^{\probless,\actprofzero}$ &  set of all vertices in $\orient{\hyperc}_{\nplayers}^{\probless}$ not accessible from $\actprofzero$, introduced in \cref{eq:partition}\\
$\rorient{\nolarg}_{\nplayers}^{\probless,\actprofzero}$ &  set of all vertices in $\rorient{\hyperc}_{\nplayers}^{\probless}$ not accessible from $\actprofzero$, introduced in \cref{eq:partition01}\\
$\nplayers$ & number of players, introduced in \cref{eq:game-def}\\
$\players$ & set of players, introduced in \cref{eq:game-def}\\
$\Norm(\mu,\sigma^{2})$ & normal distribution with mean $\mu$ and variance $\sigma^2$, as introduced after \cref{th:case-atoms}\\
$\equilibria(\game_{\nplayers})$ & set of \acp{PNE} in $\game_{\nplayers}$, introduced in \cref{de:Nash-equilibrium}\\
$\expproc_{\run}$ & exploration process at time $k$, defined in \cref{eq:Pk+1}\\
$\orientedges{\actprof}{\actprofalt}$ &  event in which $\orient{[\actprof, \actprofalt]} \in \orient{\edges}$, introduced in \cref{eq:ind-s-to-t}\\
$\actprof$ & a strategy profile, introduced after \cref{eq:game-def}\\
$\actprof_{-\play}$ & strategy profile $\actprof$ for all players except $\play$, introduced after \cref{eq:game-def}\\
$\actions_{\play}$ & set of strategies for player $i$, introduced in \cref{eq:game-def}\\
$\actions$ & set of all possible strategy profiles, introduced after \cref{eq:game-def}\\
$[\actprof, \actprofalt]$ &  edge connecting vertices $\actprof$ and $\actprofalt$\\
$\sequilibria(\game_{\nplayers})$ & set of \acp{SPNE} in $\game_{\nplayers}$, introduced in \cref{de:Nash-equilibrium}\\
$\orient{[\actprof, \actprofalt]}$ & edge $[\actprof, \actprofalt]$ oriented from $\actprof$ to $\actprofalt$\\
$\trap$ & subgraph of $\hyperc_{\nplayers}$ corresponding to $\orient{\trap}$, defined in the proof of \cref{th:BRD}\\
$\orient{\trap}$ & a trap, as defined in \cref{de:trap}\\
$\rorient{\trap}$ & subgraph of $\rorient{\hyperc}_{\nplayers}^{\probless}$ corresponding to $\orient{\trap}$, defined in the proof of \cref{th:BRD}\\
$\vertices_{\nplayers}$ & vertex set of $\hyperc_{\nplayers}$\\
$\vertices_{\nplayers}^{\even}$ & set of vertices whose sum of coordinates is even, introduced in \cref{eq:bipartite}\\
$\vertices_{\nplayers}^{\odd}$ & set of vertices whose sum of coordinates is odd, introduced in \cref{eq:bipartite}\\
$\incid_{\nplayers}$ & number of vertices incident to at least $2^{\nplayers - \floorm_{\probless}}$ unoriented vertices\\
$\rand_{\play}^{\actprof}$ & random variable dictating the payoff for player $\play$ of strategy profile $\actprof$\\
$\probeq = 1 - 2\probless$ &  probability of payoffs being equal, introduced in \cref{eq:probeq}\\
$\game_{\nplayers}$ & game with $\nplayers$ players, introduced in \cref{eq:game-def}\\
$\boundary \verticesalt$ &  set of vertices which are neighbors of vertex set $\verticesalt$, introduced in \cref{eq:neigh}\\
$\upbound_{\nplayers}$ & a bounding constant, defined in \cref{eq:hndef}\\
$\orient{\boundary \verticesalt}$ &  set of vertices which are out-neighbors of vertex set $\verticesalt$, introduced in \cref{eq:out-neigh}\\
$\equi{\actprofalt}$ & the indicator of the event that the vertex $\actprofalt$ is incident only to unoriented edges  in $\orient{\hyperc}_{\nplayers}^{\probless}$, introduced in \cref{eq:sum-isol}\\
$\equii{\nplayers}$ & defined in \cref{eq:sum-isol}\\
$\kolmo$ & Kolmogorov distance, defined in \cref{eq:Kolmogorov}\\
$\ecomp_{\ncomp}$ & $\Expect\bracks{\comp_{\ncomp}}$\\
$\rgame_{\nplayers}$ & random game with $\nplayers$ players\\
$\dist_{\TV}$ & total variation distance, defined in \cref{eq:TV}\\
$\isoln_{\nplayers}$ & the set of all isolated vertices in $\percol^{\probless}_{\nplayers}$, introduced in \cref{pr:largest-connected-component}\\
$\actprofzero$ &  strategy profile $(0,0,...,0)$\\
\end{longtable}

\end{document}